\documentclass[11pt,english]{iopart} 

\usepackage[usenames,dvipsnames]{color}
\usepackage[normalem]{ulem}
\usepackage{graphicx}

\expandafter\let\csname equation*\endcsname\relax
\expandafter\let\csname endequation*\endcsname\relax

\usepackage{amsmath}
\usepackage{amssymb,amsthm}
\usepackage{amsfonts}
\usepackage{color}
\usepackage{graphics}
\usepackage{epsfig}
\usepackage{bbm}
\usepackage{pstricks}
\usepackage{subfigure}
\usepackage{bm}
\usepackage{bbold}
\usepackage{cite}

\newtheorem{thm}{Theorem}
\newtheorem{lem}[thm]{Lemma}
\newtheorem{cor}[thm]{Corollary}

\newcommand{\be}{\begin{equation}}
\newcommand{\ee}{\end{equation}}
\newcommand{\eea}{\end{eqnarray}}
\newcommand{\bea}{\begin{eqnarray}}

\newcommand{\mean}[1]{\ensuremath{\langle{#1}\rangle}}

\newcommand{\eins}{\mathbb{1}}

\newcommand{\WW}{\ensuremath{\mathcal{W}}}

\newcommand{\MM}{\ensuremath{\mathcal{M}}}

\newcommand{\ketbra}[1]{\ensuremath{| #1 \rangle \!\langle #1 |}}

\newcommand{\ket}[1]{\ensuremath{|#1\rangle}}

\newcommand{\bra}[1]{\ensuremath{\langle#1|}}

\newcommand{\braket}[2]{\ensuremath{\langle #1|#2\rangle}}

\newcommand{\kommentar}[1]{}

\newcommand{\vr}{\ensuremath{\varrho}}

\newcommand{\forget}[1]{}

\newcommand{\lueq}{\ensuremath{\simeq_{\mathrm{LU}}}} 
\newcommand{\ptr}[1]{\tr_{#1}}                                
\newcommand{\diag}{\ensuremath{\mathop{\mathrm{diag}}}}       
\newcommand{\avg}[1]{\ensuremath{\left\langle#1\right\rangle}}       
\newcommand{\supp}{\ensuremath{\mathop{\mathrm{supp}}}}

\begin{document}


\title{Entanglement and nonclassical properties of hypergraph states}

\author{Otfried G\"uhne$^1$, Mart\'{\i} Cuquet$^2$, Frank E.S. Steinhoff$^1$,\\
Tobias Moroder$^1$, Matteo Rossi$^3$, 
Dagmar Bru\ss$^4$, \\
Barbara Kraus$^2$, and Chiara Macchiavello$^3$}
\address{$^1$Naturwissenschaftlich-Technische Fakult\"at,
Universit\"at Siegen,
Walter-Flex-Stra{\ss}e~3,
57068 Siegen, Germany}

\address{$^2$Institut f\"ur Theoretische Physik, 
Universit\"at Innsbruck, Technikerstra{\ss}e 25, 
\\
6020 Innsbruck, Austria}

\address{$^3$Dipartimento di Fisica and INFN-Sezione di Pavia, 
Via Bassi 6, 27100 Pavia, Italy
}

\address{$^4$Institut f\"ur Theoretische Physik III, 
Heinrich-Heine-Universit\"at D\"usseldorf,
\\
40225 D\"usseldorf, Germany}

\date{\today}


\begin{abstract}
Hypergraph states are multi-qubit states that form a subset of 
the locally maximally entangleable states and a generalization of 
the well-established notion of graph states.  Mathematically, they
can conveniently be described 
by a hypergraph that indicates a possible generation procedure of these 
states; alternatively, they can also be phrased in terms of a non-local 
stabilizer formalism. In this paper, we explore the entanglement properties 
and nonclassical features of hypergraph
states. First, we identify the equivalence classes under local unitary
transformations for up to four qubits, as well as important classes
of five- and six-qubit states, and determine various entanglement properties
of these classes. Second, we present general conditions under which the local 
unitary equivalence of hypergraph states can simply be decided by considering
a finite set of transformations with a clear graph-theoretical interpretation.
Finally, we consider the question whether hypergraph states and their 
correlations can be used to reveal contradictions with classical hidden 
variable theories. We demonstrate that various noncontextuality inequalities
and Bell inequalities can be derived for hypergraph states. 
\end{abstract}




\section{Introduction}

The study of multiparticle entanglement has attracted much attention 
in the last years. From the theoretical side,  multiparticle 
entanglement may be a key element to improve various applications like 
quantum information processing or quantum metrology, or to understand
and simulate physical systems, such as quantum spin chains undergoing 
a quantum phase transition. From the experimental side, the generation 
and certification of the various interesting multiparticle states  
poses tremendous challenges, but it offers the opportunity 
to demonstrate the advances in controlling and manipulating physical 
systems at the quantum level. 

Since the dimension of the Hilbert space grows exponentially with the number
of particles, it has turned out to be fruitful 
to identify families of multiparticle states which allow for a simple 
description with few parameters. Such a simple description may, for 
instance, originate from the symmetries of the states under 
scrutiny. In fact, many cases are known where various symmetries 
allow the solution of problems in quantum information theory that 
are not yet solved in the general case.

One interesting family of multi-qubit states that has attracted a 
lot of attention in the last ten years are the so-called graph 
states \cite{graphs1, hein}.  These states are of importance in various 
applications of quantum information processing, such as 
measurement-based quantum computation and  quantum error correction. 
Apart from their relevance for applications, they have a simple 
description in terms of graphs: Starting from an arbitrary graph 
(that is, a set of vertices with edges connecting them), one can 
generate the graph states from a product state by applying  
entangling quantum operations on the connected vertices. In addition, 
graph states can also be described as eigenstates of a set of 
commuting local observables, the so-called stabilizer. Due to 
their importance and their simple mathematical description, the 
entanglement properties of graph states have intensively been 
studied. For instance, different entanglement classes of graph 
states up to eight qubits have been identified \cite{hein, adan8}, 
their entanglement properties have been discussed 
\cite{VandenNest2004, Ji2010_LU-LC, PhysRevA.84.052319} 
and purification protocols have been provided \cite{graphpuri, Carle2013}. 
The stabilizer formalism for graph states has turned out to be a very useful 
tool to develop Bell inequalities or Kochen-Specker arguments 
\cite{GHZargument,  Mermin93, DP97}. Interestingly, this was 
known already long before the mathematical 
formulation of graph states was given.

A possible generalization of graph states are the locally maximally
entangleable (LME) states.  Their notion goes back to the study for which
multi-qubit states there is a local interaction with local auxiliary systems,
so that after this operation all the auxiliary systems are maximally entangled
with the initial qubits. States with this property are called LME states, and
they have been characterized in Ref.~\cite{Kruszynska2009}. It has turned out 
that they can be
generated similarly to graph states from a product state with simple
interactions (parameterized with a phase $\varphi$), but in this case also
interactions between three or more particles are needed. Mathematically, they
are distinguished by the fact that they still can be described by a stabilizer
of commuting hermitian observables. In contrast to the usual graph states,
however, the stabilizer observables are non-local, and not simple tensor
products of Pauli matrices. The usual
graph states mentioned above belong to the family of LME states. For 
them, only two-qubit interactions with a phase $\varphi=\pi$ are required.
The states with multi-qubit interactions but still a restricted phase 
of $\varphi=\pi$ are called $\pi$-LME states or hypergraph states.
As noted in Refs. \cite{Qu2013_encoding, Rossi2013} these states have 
a simple description 
in terms of more general objects than graphs, the so-called hypergraphs, 
which has motivated the 
name. In a hypergraph, one edge is allowed to connect any number of vertices,
so there are also edges with three, four, or more vertices. Hypergraph states
occur naturally in the analysis of quantum algorithms such as the
Deutsch-Jozsa algorithm and the Grover algorithm \cite{Rossi2013, scripta}.
Despite their simple and elegant description, and in contrast to graph states, hypergraph states can be complex
in the sense of Kolmogorov complexity, so they 
could, for instance, be used for quantum fingerprinting protocols \cite{mora}.
Moreover, entanglement purification protocols for these states have been
developed \cite{Carle2013}.
All these facts support the conjecture that hypergraph states
can be a good test-bed to explore the subtleties of multiparticle entanglement.

In this paper, we investigate the entanglement properties of hypergraph
states. First, we ask for which cases different hypergraphs lead to
locally equivalent hypergraph states. By locally equivalent we mean that
the two states are connected with a local unitary transformation, that is 
a local change of the basis. Consequently, all their entanglement properties
are the same. For graph states, it is already known that different graphs
can lead to locally equivalent states, and a significant effort has been
devoted to the characterization of the equivalence classes \cite{graphs1, adan8, VandenNest2004, Ji2010_LU-LC}. We extend
here this approach to hypergraph states. We first provide all  entanglement classes
for up to four qubits, as well as special classes for up to six qubits. Then, 
we derive some general conditions, under which the quest for local unitary 
equivalence can be simplified by considering local Pauli operations
only. In the second part of the paper, we ask, whether the non-local 
stabilizer formalism of hypergraph states can be used to develop arguments that 
discriminate between quantum physics and classical 
hidden-variable theories. 
First, we show that a similar argument to the one by Greenberger, Horne
and Zeilinger (GHZ) can be developed for the non-local stabilizer. Based 
on this, various novel noncontextuality inequalities and Bell inequalities
can be developed. Finally, we conclude and discuss possible extensions
of our research.

\section{Hypergraph states}
In this section we review the notions of hypergraphs, hypergraph
states and recall some of their basic properties. We stress that 
essentially all of the facts presented in this section have been shown 
before \cite{Kruszynska2009,Carle2013,Qu2013_encoding,Rossi2013, Qu2013_entropic, Qu2013_relationship}. 
However, in order to be self-contained we recall these results here.
Finally, it may be useful that some of our proofs are significantly
shorter than some of the existing proofs in the literature. 

\subsection{Properties of the controlled phase gate}
Before defining hypergraph states it is useful to explain 
some properties of the controlled phase gates, as they play 
a central role in the definition of the states.

We consider a system of $N$ qubits. We denote by $e=\{i_1, ..., i_n\}$  
a subset of $n$ qubits. Then, the controlled phase gate on the set $e$ 
is the unitary transformation given by the matrix
\be
C_e = \eins - 2 \ketbra{1 \cdots 1} 
\ee
In other words, $C_e$ is a diagonal $2^n \times 2^n$ matrix in the 
standard basis with all entries equal to $1$ except from the last one, 
which has the value $-1$. Clearly, this transformation is invariant 
under permutation of the qubits and it fulfills $(C_e)^2=\eins.$ For
a single qubit, $C_{\{i\}} = \sigma_z$ equals a Pauli matrix. Furthermore, 
in order to formulate general formulae later, if  $e=\emptyset$ is an 
empty set we define $C_\emptyset = -1$ as an overall negative sign. 
Finally, we will sometimes use the notation $C_{ijk}$ instead 
of $C_{\{i,j,k\}}.$ 

For our purposes it is important to understand the relation of 
the controlled phase gate to the Pauli matrices. Here and in 
the following, we denote by $X_k, Y_k, Z_k$ the Pauli matrices 
$\sigma_x, \sigma_y, \sigma_z$ acting on the $k$-th qubit. The 
following Lemma presents two simple rules, which are relevant for 
many of the further calculations.

\begin{lem}\label{lem:rules}
If $k \in e$, the following identities hold:
\bea
X_k C_e X_k &=& C_e C_{e\setminus \{k\}} 
\label{rule1}
\\
C_e X_k C_e &=& X_k \otimes C_{e\setminus \{k\}} 
\label{rule2}
\eea
\end{lem}
In the following, we will refer to Eq.~(\ref{rule1}) as the 
{\it first rule} and to Eq.~(\ref{rule2}) as the {\it second 
rule}.
\begin{proof}
For Eq.~(\ref{rule1}) we can assume that $e=\{1,2,3,...,n\}$
and $k=1$. Then we have 
$C_e=\ketbra{0}\otimes \eins + \ketbra{1} \otimes C_{e\setminus \{k\}} $
and it follows that 
\begin{align}
X_k C_e X_k 
&= 
\ketbra{1}\otimes \eins + \ketbra{0} \otimes C_{e\setminus \{k\}}
=
\eins - 2 \ketbra{0 1 \cdots 1}
\nonumber
\\
&=
(\eins - 2 \ketbra{11 \cdots 1}) 
(\eins - 2 \ketbra{01 \cdots 1}- 2 \ketbra{11 \cdots 1})
\nonumber
\\
&= C_e (\eins \otimes C_{e\setminus \{k\}}).
\end{align}
For Eq.~(\ref{rule2}) we can apply the first rule to see that 
$C_e X_k C_e = C_e X_k C_e X_k X_k = C_e C_e C_{e\setminus \{k\}} X_k
=X_k \otimes C_{e\setminus \{k\}}.$
\end{proof}

A more general commutation rule is the following:

\begin{lem}
  Let $e$ be a set of qubits, $K\subseteq e$ a subset of $e$, and
  $\mathcal{P}(K)$ the power set of $K$ (that is, the set of all 
  subsets of $K$, including the empty set).
  Then we have
  \begin{equation}
    C_e \Big( \bigotimes_{k\in K}X_k \Big)
    = \Big( \bigotimes_{k\in K}X_k \Big) 
    \Big(\prod_{f\in\mathcal{P}(K)} C_{e\setminus f}\Big)
    \label{lemma2}
  \end{equation}
Note that for $K = e$ the operator $C_{e\setminus K}=C_\emptyset=-1$ appears,
so a sign occurs. Also, for $K=e=\{k\}$ this is the usual commutator relation
$Z_kX_k = - X_k Z_k$.
\end{lem}

\begin{proof}
The proof works by induction in the number of elements in $K.$ If 
$K=\{k\}$ consists of only one element then  Eq.~(\ref{lemma2})
is nothing but a reformulation of Eq.~(\ref{rule1}). The induction 
step is straightforward.
\end{proof}

\subsection{Hypergraphs and hypergraph states}

Let us first explain the notion of hypergraphs. A hypergraph 
$H=\{V,E\}$ consists of a set $V=\{1, \dots , N\}$ of $N$ vertices 
and a set $E$ of edges connecting the vertices. In a usual graph, 
an edge connects only two vertices, that is, any $e \in E$ is of 
the type $e=\{i,j\}$ with $i,j \in V$. For hypergraphs, however, 
also edges connecting any number of vertices are allowed and a 
general edge is just a subset of $V$. 
Note that the set of all possible hypergraphs is significantly 
larger than the set of all graphs: for $N$ vertices, there are $2^{N(N-1)/2}$
standard graphs, but there are $2^{2^N}$ hypergraphs. This implies that
hypergraph states can be complex in the sense of Kolmogorov complexity.
An example of a hypergraph is given in Fig.~\ref{hg-bild-1}.

\begin{figure}[t]
\begin{center}
\includegraphics[width=0.30\textwidth]{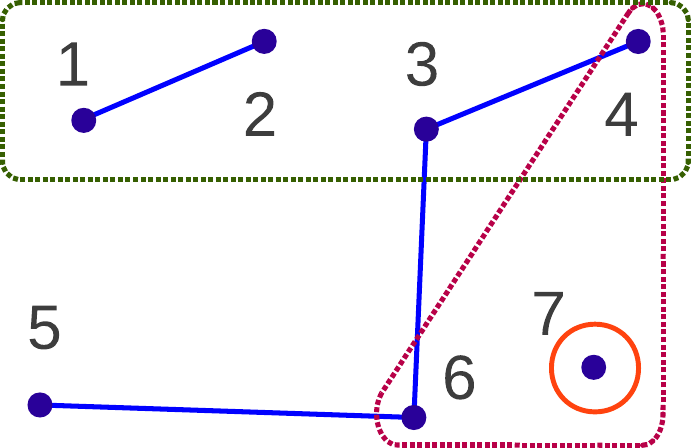}
\end{center}
\caption{Example of a hypergraph with seven vertices. The graph has
one one-edge $\{7\}$, four two-edges $\{1,2\},\{3,4\},\{3,6\},$ and $\{5,6\},$ one three-edge $\{4,6,7\}$, and one four-edge $\{1,2,3,4\}$. 
}
\label{hg-bild-1}
\end{figure}

Some more terminology is in order here. The cardinality $k$ of an 
edge is the number of vertices within this edge; we also speak of 
a $k$-edge. We call a hypergraph {uniform} of cardinality $k$, or 
$k$-uniform, if all edges are $k$-edges. Furthermore, the neighborhood 
of a vertex $i$ consists of all vertices which are connected with $i$ 
via some edge. Finally, we denote by $E(i)$ the set of all edges 
$e \in E$ with $i \in e$.

As for graph states, an $N$-qubit hypergraph state $\ket{H}$ can be 
defined by association to a hypergraph of $N$ vertices via the 
interaction history. Indeed, if one starts with a product state
$\ket{\psi}=\ket{+}^{\otimes N}$ (where $\ket{+}=(\ket{0}+\ket{1})/\sqrt{2}$
is the eigenstate of $X$), and applies all the non-local unitaries 
$C_e$ for all $e \in E$, one arrives at the hypergraph state,
\be
\ket{H} = \prod_{e \in E} C_{e} \ket{+}^{\otimes N}.
\ee
That is, a hypergraph state is associated to a hypergraph, where each
qubit corresponds to a vertex and each interaction $C_e$ is present 
if the hyperedge $e$ belongs to the associated hypergraph.

As in the case of graph states, a hypergraph state can also be defined 
via the stabilizer formalism. Starting from the stabilizer $X_i$ of 
$\ket{+}_i$, and using that $C_e^2=\eins$ and the rules in 
Lemma~\ref{lem:rules}, one can define for each vertex $i$ a
stabilizing operator $g_i$ as
\be
g_i \equiv
\big( \prod_{e\in E} C_e \big) X_i \big( \prod_{e\in E} C_e \big)
=
\big(\prod_{e\in E(i)} C_e \big) X_i \big( \prod_{e\in E(i)} C_e \big)
=
X_i \otimes \big( \prod_{e \in E(i)} C_{e \setminus \{i\} }\big)
\label{gendef}
\ee
{From} this definition one immediately sees that the stabilizers commute, 
since $g_i g_j = \big(\prod_{e \in E} C_{e} \big) X_i X_j \big( \prod_{e \in E} C_{e}\big)
=
\left( \prod_{e \in E} C_{e} \right) X_j X_i \left( \prod_{e \in E} C_{e}\right)
=
g_j g_i
$.
The $g_i$ are traceless and form a maximal set of $N$ commuting observables with 
eigenvalues $\pm 1$, hence there must be a common eigenbasis. More directly, the
hypergraph state $\ket{H}$ can be defined as the common eigenstate to 
all $g_i$ with the eigenvalue $+1$,
\be
g_i \ket{H} = \ket{H} \quad \mbox{ for all } i.
\label{hyperdef}
\ee
The other states in the eigenbasis can be defined by taking other 
possibilities as eigenvalues. The resulting $2^N$ states are all 
orthogonal, and form the so-called hypergraph-state basis, 
which can be generated by successive application of Pauli $Z$ transformations, that is
$\ket{H_{\bm{k}}}\equiv\ket{H_{k_1,\dots,k_N}}=Z_1^{k_1}\otimes\dots\otimes
Z_N^{k_N}\ket{H}$.

The definition of the stabilizing operators can be used to define 
the stabilizer as a group. Let us consider the set of $2^N$ 
observables
\be
\mathcal{S} = \{S_{\bm{x}} | S_{\bm{x}} =\prod_{i \in V} (g_i)^{x_i} 
\mbox{ with } \bm{x}\in\mathbb{Z}_2^N \},
\ee
which consists of all products of the operators $g_i$. $\mathcal{S}$ 
is an Abelian group with $2^N$ elements,  and for any element 
we have $S_{\bm{x}} \ket{H} = \ket{H}$. Furthermore, we can 
express the hypergraph state as
\be
\ketbra{H} = \frac{1}{2^N} \sum_{\bm{x}\in\mathbb{Z}_2^N} S_{\bm{x}} = \prod_{i=1}^{N}\frac{g_i+\eins}{2}.
\ee
These formulas can directly be verified by using the fact that any 
product of the $g_i$ is diagonal in the hypergraph-state basis.
Such a result is, of course, well known for graph states. 
The only difference here is that for hypergraph states the stabilizer contains
also nonlocal observables.

Before proceeding to a general discussion of the properties of hypergraph 
states it is useful to note that the set of edges $E$ can also contain
single vertex edges of the type $e= \{k\},$ but also the empty edge 
$e=\emptyset$. Such contributions will occur in the formulae in
forthcoming discussion. These edges are, however, not relevant for  the 
entanglement properties of the hypergraph state: An edge $e= \{k\}$ 
corresponds to a local unitary transformation $C_k =Z_k$ on the $k$-th
qubit and the empty edge $e=\emptyset$ induces a global sign shift on the
state $\ket{H}$. Both transformations do not change any entanglement 
properties, so we will usually neglect the single-vertex edges and empty
edge in the following. 

As mentioned before, hypergraph states are a special case of the more general
class of LME states.
These states are defined by $2^N$ real phases, and can be prepared by applying 
general phase operators $\tilde{C}_e(\varphi) = \eins - (1-e^{i \varphi})
\ketbra{1 \cdots 1}$ to the state $\ket{+}^{\otimes N}$. As in the 
case of graph states, LME states can be associated to a weighted 
hypergraph, where the weight of each edge corresponds to the phase 
of its associated operator. Setting the phases to $\pi$ 
results in the class hypergraph states or $\pi$-LME states.

\subsection{Local Pauli transformations and graph transformations
\label{sec:graph_trafos}}

In this section we will discuss the possible actions of Pauli matrices
as local unitary transformations on a hypergraph state. First, we have 
mentioned already the action of the unitary transformation $Z_k$ on some
qubit $k$. It adds the edge $e= \{k\}$  to the set $E$ if this edge
was not yet contained in $E$. If, however, the edge $e= \{k\}$
was already in $E$, the transformation $Z_k$ removes this edge
again, since $(C_k)^2=\eins$.

\begin{figure}[t]
\begin{center}
\includegraphics[width=0.60\textwidth]{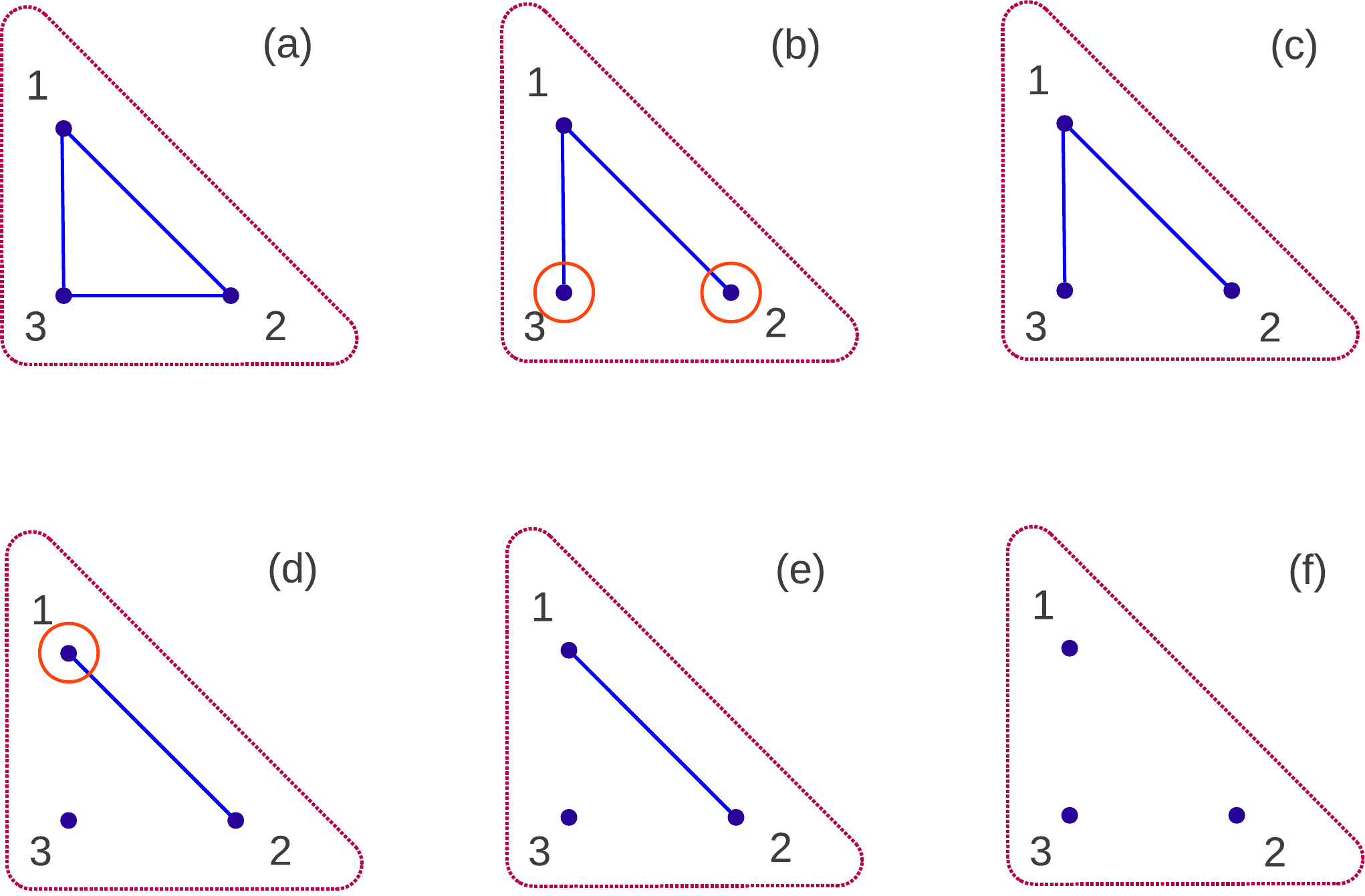}
\end{center}
\caption{Example of the different local Pauli operations and their
action on the hypergraph. 
(a$\rightarrow$b) The operation $X_1$ on the first qubit.
(b$\rightarrow$c) The two one-edges are removed by $Z$ operations
on the qubits 2 and 3.
(c$\rightarrow$d) The operation $X_2$ on the second qubit.
(d$\rightarrow$e) The one-edge is removed by a $Z_1$ operation.
(e$\rightarrow$f) The operation $X_3$ on the third qubit.
}
\label{hg-bild-2}
\end{figure}
A more interesting transformation is the application of the unitary matrix
$X_k.$ We can directly calculate that
\begin{equation}
X_k \ket{H} =
X_k \big( \prod_{e \in E} C_{e} \big) \ket{+}^{\otimes N}
= \big( \prod_{e \in E} C_{e} \big)
\big( \prod_{e \in E(k)} C_{e\setminus \{k\}} \big) \ket{+}^{\otimes N}.
\end{equation}
This is again a hypergraph state. Its edges can be conveniently 
described, if one considers the symmetric difference of sets. 
Generally, the symmetric difference of two sets $A$ and $B$ is given 
by 
\be
A \Delta B = (A\cup B) \setminus (A \cap B),
\ee
that is, both sets are joined, but then the elements which are 
in both sets are removed. Let us define the set of edges
\be
E^{(k)} = \{e \setminus \{k \} | e\in E(k) \}.
\ee
This set is formed by first taking all edges in $E$ which contain
$k$, and then removing $k$ out of all these edges. The set of
edges of the state $X_k \ket{H}$ is then given
by
\be
E^{\rm new} = E \Delta E^{(k)}.
\ee
The hypergraph of the state $X_k \ket{H}$ can directly be determined 
graphically: One picks the vertex $k$ and determines the set $E^{(k)}.$
Then one removes or adds these edges to $E$, depending on whether they 
exist already in $E$ or not. This is demonstrated in an example in 
Fig.~\ref{hg-bild-2}. Finally, the unitary transformation $Y_k$ can 
be implemented by first applying $X_k$ and then $Z_k$.

In the following, we will investigate which hypergraph states are 
equivalent under local unitary transformations, especially local 
Pauli operations. We will show that under certain constraints local unitary
(LU) equivalence is the same as equivalence under local Pauli transformations.
It can be seen already from the present discussion that
some properties of the hypergraph remain invariant under local Pauli 
operators, so they can be used to distinguish equivalence classes. The 
main property concerns the edges with the highest cardinality, that is, the
edges which connect the largest number of vertices. These remain 
invariant under $X_k$ transformations, since the set $E^{(k)}$ cannot
contain edges with maximal cardinality. From this, it immediately follows that
any $k$-uniform hypergraph state cannot be Pauli-equivalent to another 
$k'$-uniform hypergraph, unless they are identical. This fact was noted 
for the case $k\neq k'$ before; the proof in Ref.~\cite{Rossi2013} required, 
however, lengthy calculations. Moreover, if one has a $k$-edge $e$ in a hypergraph, and
a $(k-1)$-edge $e'$ within this edge ($e' \subset e$), then one can always remove
$e'$ with a local Pauli operation (but this may introduce new edges of cardinality
$k-1$ since there may be other $k$-edges apart from $e$).
For the special case that an $N$-qubit hypergraph state contains an $N$-edge, one
can remove all the $(N-1)$-edges. This explains why the hypergraphs no.~17-27
in Fig.~\ref{hg-bild-vier} do not contain any 3-edges.

\section{Entanglement classes}
In this section, we distinguish and investigate the different 
LU equivalence classes for hypergraph states. Two
states are LU equivalent, and belong to the same class, if they are
connected by local unitary transformations. That is, if there exist 
local unitary operators $U_1,\dots,U_N$ such that
\begin{equation}
  \ket{\psi} = U_1\otimes\dots\otimes U_N\ket{\phi}.
  \label{eq:LUequiv_def}
\end{equation}
In this case we also write $\ket{\psi}\lueq\ket{\phi}$. We stress 
that other equivalence transformations are worth to be considered, 
for instance stochastic local operations and classical communication 
(SLOCC) \cite{gtreview}. For the present paper, however, we focus on 
equivalence under local unitaries and local Pauli operations, as the 
latter have a clear graph-theoretical interpretation.

In our classification, we focus on hypergraph states which have at least one 
edge with three or more vertices. The reason is that the states with only 
two-edges are graph states and their entanglement classes have already been 
extensively discussed \cite{graphs1, adan8, VandenNest2004, Ji2010_LU-LC}.

\subsection{Three qubits}
For three qubits, there is only one LU equivalence class of hypergraph states. It can 
be represented by the hypergraph in Fig.~\ref{hg-bild-2}(f). This figure also 
shows that any other three-qubit hypergraph state with a three-edge 
can be transformed into this state by local Pauli transformations. 
Despite this fact has been noted before \cite{Qu2013_encoding}, we will
discuss it in some detail, since this allows us to introduce the 
concepts and quantities that are also used for the general case. 

This hypergraph state is given by
\begin{align}
\ket{\hat{H}_3} & = \frac{1}{\sqrt{8}}(\ket{000}+\ket{001}+\ket{010}+\ket{011}+\ket{100}+
\ket{101}+\ket{110}-\ket{111}),
\end{align}
but after a Hadamard transformation on the third qubit one arrives at 
the simpler form
\begin{align}
\ket{H_3} & = \frac{1}{2}(\ket{000}+\ket{010}+\ket{100}+\ket{111}).
\end{align}
The entanglement properties of this state are the following: The reduced 
single-qubit matrices have all the largest eigenvalue equal to $3/4 $ and, 
consequently, the second eigenvalue is $1/4$. A possible entanglement 
witness detecting the entanglement in a general hypergraph state is 
of the type
\be
\WW = \alpha_S \eins - \ketbra{H},
\ee
where $\alpha_S$ is the maximal squared overlap between the state 
$\ket{H}$ and any biseparable pure state, that is a state of the form 
$\ket{\phi}= \ket{\alpha}_M \ket{\beta}_{\bar{M}}$ where $M |\bar{M}$ 
is some bipartition of the $N$ qubits. This overlap $\alpha_S$ can be 
directly computed as the maximum of all eigenvalues of all reduced 
states, so we have for the three-qubit state $\alpha_S = 3/4.$
In general, one can show that for $N$-qubit hypergraph states consisting
of only a single $N$-edge, one has that $\alpha_S=1-2^{(1-N)}$ \cite{scripta}.

The entanglement in the state $\ket{H_3}$ can be quantified by several 
entanglement monotones.  First, the bipartite negativity, defined as 
the sum of the absolute values of all negative eigenvalues, is for the $A|BC$
bipartition given by 
$N_{A|BC}(\vr) \equiv \sum_k |\lambda^-_k (\vr^{T_A})| = \sqrt{3}/4,$ 
for the other bipartitions it is the same. Second, we consider the geometric 
measure of entanglement, which is for pure states defined via
\be
E_G (\ket{\psi}) = 1- \max_{\ket{\phi}=\ket{a}\ket{b}\ket{c}}|\braket{\phi}{\psi}|^2,
\ee
as one minus the maximal squared overlap with fully product states. 
For mixed states, this entanglement monotone can be extended using the 
convex roof. Since LU equivalent states have the same geometric measure, this 
quantity can be used to prove that two states are not LU equivalent.
For the three-qubit hypergraph  state we find by direct numerical 
optimization
$
E_G ({\ket{H_3}}) = 0.32391
$.

A third measure which can be used to characterize genuine multipartite 
entanglement is the genuine multiparticle negativity, coming from the 
characterization of multiparticle entanglement using PPT mixtures 
\cite{Jungnitsch11, pptmixerprogram}. 
This measure can for general mixed states be computed via semidefinite 
programming (SDP), but for pure states and using the normalization of 
Ref.~\cite{hofmann} it is given by the minimal bipartite negativity, so in our 
case we have $N_G(\ket{H_3})=\sqrt{3}/4.$

Finally, we can investigate the robustness of the entanglement under 
noise. For that we consider states of the form
\be
\vr(p)= (1-p) \ketbra{H} + p \frac{\eins}{2^N}
\ee
and ask for the maximal value $p_{\rm max}$ for which the state is 
still genuine multiparticle entangled. Using the approach of PPT 
mixtures, a lower bound $p_{\rm ppt}$ on the maximal $p_{\rm max}$ 
can directly be computed via SDP \cite{pptmixerprogram}, and for the state $\ket{H_3}$ we 
find $p_{\rm ppt} = 0.4952.$ Two comments are in order here: First, 
$p_{\rm ppt}$ is clearly invariant under local unitaries, so it can 
be used to distinguish different LU classes. It is better suited 
for this task than the geometrical measure $E_G$, since for 
$p_{\rm ppt}$ the SDP results in a certified solution, while the 
computation of $E_G$ may, at least in principle, have
the problem of not finding the global optimum. Second, since 
for $\ket{H_3}$ the state $\vr(p)$ is locally equivalent to a 
permutationally invariant three-qubit state, the PPT mixture 
condition is  a necessary and sufficient condition for separability 
\cite{pipptmixer}.
So we have $p_{\rm max}=p_{\rm ppt}$ 
and states $\vr(p)$ with $p > 0.4952$ are biseparable.

Finally, note that for the three-qubit fully entangled hypergraph 
states without three-edge, that is, the usual graph states, there 
is one further equivalence class. This is given by the fully connected 
graph with the edges $E=\{\{1,2\},\{2,3\}, \{1,3\}\}$ and in a suitable 
basis the representing state is the three-qubit GHZ state
$
\ket{GHZ_3} = \tfrac{1}{\sqrt{2}}(\ket{000}+\ket{111}).
$

\begin{figure}[t]
\begin{center}
\includegraphics[width=0.7\textwidth]{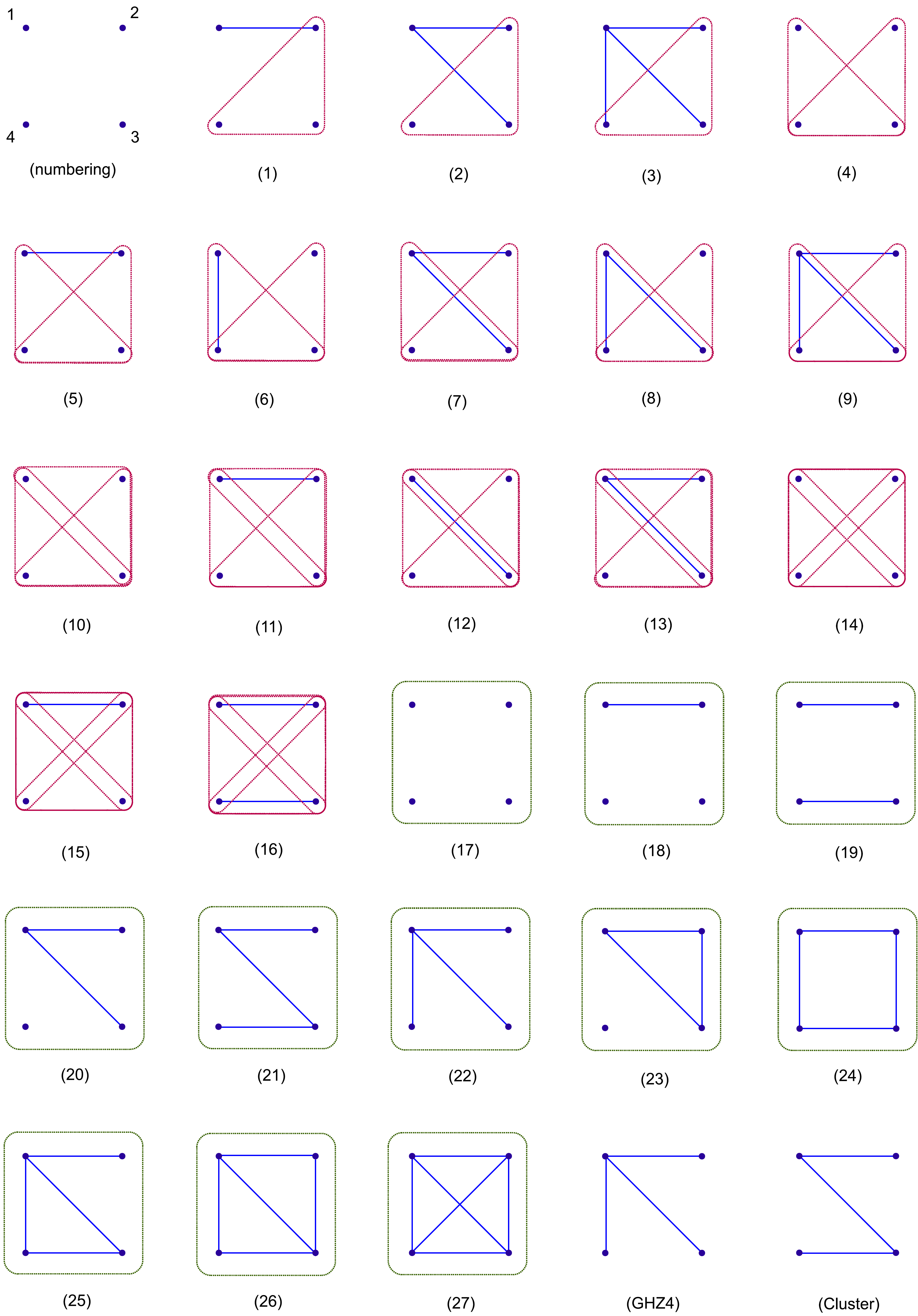}
\end{center}
\caption{Representatives of all the 27 equivalence classes 
under local unitary transformations for hypergraph states of four
qubits. In addition, the usual graphs of the four-qubit GHZ state
and the four-qubit cluster state are shown.
}
\label{hg-bild-vier}
\end{figure}

\subsection{Four qubits}

In order to characterize all the entanglement classes for four 
qubits we proceed as follows: We investigate all the $2^{2^4}=65536$ 
possible hypergraphs and test for equivalence under local Pauli 
operations as well as under permutation of the qubits. After that, 
we find 27 classes among all hypergraphs which contain at least 
one three-edge and where all vertices are connected. These 
hypergraphs are shown in Fig.~\ref{hg-bild-vier}.

In Table \ref{hg4table} we list the entanglement properties of the
corresponding hypergraph states. From the properties of the 
reduced density matrices and the reported values of $E_G$ and 
$p_{\rm ppt}$ one can directly conclude that these states are 
inequivalent under general local unitary transformations and not
only under local Pauli transformations. 

Three states among the 27 states are of special interest. These are
the states no.~3, no.~9 and no.~14, as for these states the reduced
single qubit matrices are maximally mixed. Therefore, they are all in the
maximally entangled set~\cite{DeVicente2013_maximally}, i.e.\ they
cannot be reached by transformations consisting of local operations and
classical communication from any other state with the same
number of qubits. The state no.~3 can, after applying a Hadamard 
transformation on the first qubit, be  written as:
\begin{align}
\ket{V_3} &
= \frac{1}{\sqrt{8}}\big[(\ket{0011}+\ket{0101}+\ket{0110}+\ket{1001}+
\ket{1010}+\ket{1100})+(\ket{0000}-\ket{1111})\big]
\nonumber
\\
& = \sqrt{\frac{{3}}{{4}}} \ket{D_4} + \frac{1}{2} \ket{GHZ_4^-},
\end{align}
where $\ket{GHZ_4^-}=(\ket{0000}-\ket{1111})/\sqrt{2}$
and $\ket{D_4}=(\ket{0011}+\ket{0101}+\ket{0110}+\ket{1001}+
\ket{1010}+\ket{1100})/\sqrt{6}.$
That is, it is a superposition of the four-qubit Dicke state
with two excitations with a four-qubit GHZ state. From this form 
it is clear that the state $\ket{V_3}$ belongs to the symmetric 
subspace. The state $\ket{V_3}$ can be viewed as if first a GHZ state
via two-body interactions is created, then, an additional controlled
phase gate on the qubits $\{2,3,4\}$ is applied. This still increases
the entanglement, since the geometric measure as well as the robustness
$p_{\rm ppt}$ is higher than for a GHZ state. 

The second state, no.~9, can (after applying 
$X_3$ and $Z_1$ and further Hadamard transformations on qubits 3 and 4)
be written as
\begin{align}
\ket{V_9} &
= \frac{1}{\sqrt{2}} \ket{GHZ_4^-}
+\frac{1}{2} \ket{01}\ket{\gamma} + \frac{1}{2} \ket{10}\ket{\overline{\gamma}},
\end{align}
where
$\ket{\gamma}=(\ket{00}+\ket{01}-\ket{10}+\ket{11})/{2}$
and 
$\ket{\overline{\gamma}}=(\ket{00}-\ket{01}+\ket{10}+\ket{11})/{2}$ are Bell-type states.
{From} this form, a symmetry under simultaneous permutation of the qubits
$1 \leftrightarrow 2$ and $3 \leftrightarrow 4$ becomes apparent.

The third state with the property that all single 
qubit reduced states are $\mathbb{1}/2$ is the state no.~14. After a Hadamard 
and a $Z$ transformation on the fourth qubit this is given by
\begin{align}
\ket{V_{14}} &
= \frac{1}{\sqrt{8}}\big[(\ket{0011}+\ket{0101}+\ket{0110}+\ket{1001}+
\ket{1010}+\ket{1100})+(\ket{0001}-\ket{1110})\big]
\nonumber
\\
& = \sqrt{\frac{{3}}{{4}}} \ket{D_4} + \frac{1}{2} \ket{\overline{GHZ}_4^-},
\end{align}
where $\ket{\overline{GHZ}_4^-}=(\ket{0001}-\ket{1110})/\sqrt{2}$,
so it is a superposition of a Dicke state with another GHZ state.
Note that in a suitable basis this state belongs to the symmetric 
subspace again, as can be seen from the symmetry of the hypergraph.

Besides these highly entangled states, seven hypergraph states have a very simple
form, if they are written in the appropriate basis. These are
\begin{align}
\ket{V_{1}} &=\frac{1}{2}(\ket{0000}+\ket{0001}+\ket{1100}+\ket{1111}),
\nonumber
\\
\ket{V_{2}} &=\frac{1}{2}(\ket{0000}+\ket{0111}+\ket{1010}+\ket{1100}),
\nonumber
\\
\ket{V_{4}} &=\frac{1}{2}(\ket{0000}+\ket{0001}+\ket{0010}+\ket{1111}),
\nonumber
\\
\ket{V_{6}} &=\frac{1}{2}(\ket{0000}+\ket{0010}+\ket{0111}+\ket{1001}),
\nonumber
\\
\ket{V_{8}} &=\frac{1}{2}(\ket{0000}+\ket{1001}+\ket{1010}+\ket{1111}),
\nonumber
\\
\ket{V_{10}} &=\frac{1}{\sqrt{8}}(2\ket{0000}+\ket{0011}+\ket{0110}+\ket{1010}-\ket{1111}),
\nonumber
\\
\ket{V_{12}} &=\frac{1}{\sqrt{8}}(2\ket{0000}+\ket{0010}-\ket{0111}+\ket{1011}+\ket{1110}).
\end{align}

In addition to the 27 classes, there are two LU equivalence classes
for the usual graph states, the GHZ and the cluster state. The corresponding 
graphs are also shown in Fig.~\ref{hg-bild-vier}. Finally, it is worth 
mentioning that some other well-known states for four qubits \cite{gtreview}, 
such as the 
W state $\ket{W_4}$, the four-qubit singlet state $\ket{\Psi_4}$ or
the $\chi$-state $\ket{\chi_4}$ are not locally equivalent to any 
hypergraph state.

\begin{table}[t]
\begin{center}
\begin{tabular}{|c|c| c c c c | c c c| c | c| c| }
\hline
Class & $E_G$ & $\varrho_A$ & $\varrho_B$ & $\varrho_C$ & $\varrho_D$ 
& $\varrho_{AB}$ & $\varrho_{AC}$ & $\varrho_{AD}$ & $\alpha_{BS}$ & 
$N_{\rm gen}$ & $p_{\rm ppt}$
\\
\hline
\hline
1 & 0.50000 & 1/2 & 1/2 & 3/4 & 3/4 & 3/4 & 1/2 & 1/2 & 3/4 & $\sqrt{3}/4$& 0.5430
\\
2 & 0.65651 & 1/2& 1/2&1/2 &3/4 &1/2 &1/2 &1/2 & 3/4& $\sqrt{3}/4$&0.5803
\\
3 &0.65277 & 1/2 &  1/2&  1/2&  1/2&  1/2& 1/2 &  1/2&  1/2& 1/2&0.5664
\\
4 & 0.34549& 3/4& 3/4& 3/4&3/4 & 3/4& $\Gamma_1$&$\Gamma_1$ & 3/4& $\sqrt{3}/4$&0.5286
\\
5 & 0.57322 & 1/2&1/2 &3/4 &3/4 &3/4 &$\Gamma_2$ &$\Gamma_2$ & 3/4&$\sqrt{3}/4$&0.5549
\\
6 & 0.50000 & 3/4& 3/4& 3/4&1/2 &1/2 &$\Gamma_1$&$\Gamma_1$ & 3/4& $\sqrt{3}/4$&0.5466
\\
7 & 0.62500 & 1/2&1/2 &1/2 &3/4 &1/2 &$\Gamma_2$ & $\Gamma_2$&3/4 & $\sqrt{3}/4$&0.5722
\\
8 & 0.63572 & 3/4&3/4 & 1/2&1/2 &1/2 & $\Gamma_1$& $\Gamma_1$& 3/4&$\sqrt{3}/4$ &0.5409
\\
9 & 0.63572 & 1/2&1/2 &1/2 &1/2 & 1/2 & $\Gamma_2$& $\Gamma_2$&1/2 & 1/2& 0.5848
\\
10 & 0.50000 &3/4 & 3/4& 1/2&3/4 & $\Gamma_1$& $\Gamma_1$& $\Gamma_1$&3/4 & $\sqrt{3}/4$ &0.5306
\\
11 & 0.59872 & 1/2&1/2 &1/2 &3/4 & $\Gamma_1$& $\Gamma_2$& $\Gamma_2$&3/4 & $\sqrt{3}/4$&0.5700
\\
12 & 0.37500&3/4 & 3/4&3/4 &3/4 &$\Gamma_1$ &$\Gamma_1$ & $\Gamma_1$&3/4 &$\sqrt{3}/4 $ & 0.5242
\\
13 & 0.62500& 1/2 &1/2 & 3/4& 3/4& $\Gamma_1$& $\Gamma_2$& $\Gamma_2$&3/4 & $\sqrt{3}/4$&0.5523
\\
14 & 0.57161 & 1/2& 1/2& 1/2& 1/2& $\Gamma_1$&$\Gamma_1$ &$\Gamma_1$ &$\Gamma_1$ & 1/2&0.5346
\\
15 &  0.58726&  3/4& 3/4&1/2 &1/2 & $\Gamma_1$&$\Gamma_2$ &$\Gamma_2$ & 3/4& $\sqrt{3}/4$&0.5568
\\
16 & 0.43750&3/4 &3/4 &3/4 &3/4 &$\Gamma_1$ &$\Gamma_1$ & $\Gamma_1$& 3/4& $\sqrt{3}/4$&0.5157
\\
17 & 0.19018 & 7/8 &7/8 &7/8 & 7/8& $\Gamma_3$&$\Gamma_3$ &$\Gamma_3$ & 7/8& $\sqrt{7}/8$&0.4316
\\
18 & 0.43187 &5/8 &5/8 &7/8 &7/8 & $\Gamma_3$ & $\Gamma_4$ & $\Gamma_4$ & 7/8 &$\sqrt{7}/8 $&0.4806
\\
19 & 0.64376 & 5/8& 5/8& 5/8& 5/8 & $\Gamma_3$ &$\Gamma_2$ &$\Gamma_2$ & $\Gamma_3$&3/8 &0.5411
\\
20 & 0.46240 &5/8 & 5/8& 5/8&7/8 &$\Gamma_4$ &$\Gamma_4$ &$\Gamma_4$ & 7/8 &$\sqrt{7}/8$ & 0.5261
\\
21 & 0.65277 &5/8 & 5/8& 5/8&5/8 &$\Gamma_4$ & $\Gamma_2$&$\Gamma_2$ &5/8 & $\sqrt{15}/8$&0.5736
\\
22 & 0.46097 &5/8 &5/8 &5/8 &5/8 &$\Gamma_4$ &$\Gamma_4$ &$\Gamma_4$ &5/8 &$\sqrt{15}/8$ &0.5423
\\
23 & 0.54497 & 5/8&5/8 &5/8 &7/8 & $\Gamma_4$&$\Gamma_4$&$\Gamma_4$ &7/8 & $\sqrt{7}/8$&0.5193
\\
24 & 0.55656 & 5/8& 5/8& 5/8&5/8 &$\Gamma_2$ &$\Gamma_4$ & $\Gamma_2$& 5/8& $\sqrt{15}/8$&0.5744
\\
25 & 0.62926 &5/8 &5/8 & 5/8& 5/8&$\Gamma_4$ &$\Gamma_2$ &$\Gamma_2$ & 5/8& $\sqrt{15}/8$&0.5728
\\
26 & 0.53879 & 5/8& 5/8& 5/8& 5/8&$\Gamma_2$ &$\Gamma_4$ &$\Gamma_2$ & 5/8 &$\sqrt{15}/8$ &0.5699
\\
27 &0.55637 &5/8 &5/8 & 5/8&5/8 &$\Gamma_4$ & $\Gamma_4$&$\Gamma_4$ & 5/8 & $\sqrt{15}/8$&0.5229
\\
\hline
GHZ & 1/2 & 1/2& 1/2& 1/2& 1/2& 1/2& 1/2 &1/2 & 1/2 &1/2 & 4/7
\\
Cluster & 3/4 & 1/2& 1/2& 1/2& 1/2& 1/2& 1/4 &1/2 & 1/2 &1/2 & 8/13
\\
\hline
\end{tabular}
\end{center}
\caption{Entanglement properties of the 27 equivalence classes of four-qubit 
hypergraph states. The table shows the values of the geometric measure of 
entanglement $E_G$, the maximal eigenvalues of the single-qubit reduced 
states ($\vr_A, \vr_B,\vr_C,\vr_D$), the maximal eigenvalues of the 
two-qubit reduced states ($\vr_{AB}, \vr_{AC},\vr_{AD}$), 
the maximal overlap with biseparable states ($\alpha_{BS}$),
the genuine multiparticle negativity ($N_{\rm gen}$) and
the noise robustness of multiparticle entanglement according
to the PPT mixture approach ($p_{\rm ppt}$).
The used constants are given by 
$\Gamma_1=(3+\sqrt{5})/8 \approx 0.65450$,
$\Gamma_2=(2+\sqrt{2})/8 \approx 0.42677$,
$\Gamma_3=(4+\sqrt{7})/8 \approx 0.83071$,
and
$\Gamma_4= (8+13/\sqrt[3]{z}+\sqrt[3]{z})/24 \approx 0.60170$
with $z=8+i3\sqrt{237}.$
}
\label{hg4table}
\end{table}

\subsection{Five qubits}
For five qubits, there are already $2^{32}$ possible hypergraphs, 
which makes an exhaustive classification difficult. Moreover, it
can be expected that the resulting number of classes is significantly
larger than in the four-qubit case and therefore a complete classification
is of limited use.

For five and more qubits we focus therefore on $k$-uniform 
hypergraphs, since for them an exhaustive treatment is still 
possible. These type of states are arguably relevant for experiments, as their
generation requires the same type of multiqubit interaction.
As can be seen from the hypergraph transformation rule outlined in
Section~\ref{sec:graph_trafos},
any $k$-uniform hypergraph cannot be converted into a different
$k$-uniform hypergraph by local Pauli operators. We 
focus our discussion only on the $k$-uniform hypergraphs where 
for the corresponding states the reduced single-particle states 
are maximally mixed. As mentioned before, these states are relevant because
they are in the maximally entangled set \cite{DeVicente2013_maximally}, 
and, moreover, the condition of having the single-qubit reduced states 
maximally mixed is necessary for quantum error correction.
Note that for two-uniform hypergraphs (i.e.\ corresponding to graph states),
any state has the property that the reduced states are maximally mixed, 
and hence is also in the maximally entangled set.

\begin{figure}[t]
\begin{center}
\includegraphics[width=0.7\textwidth]{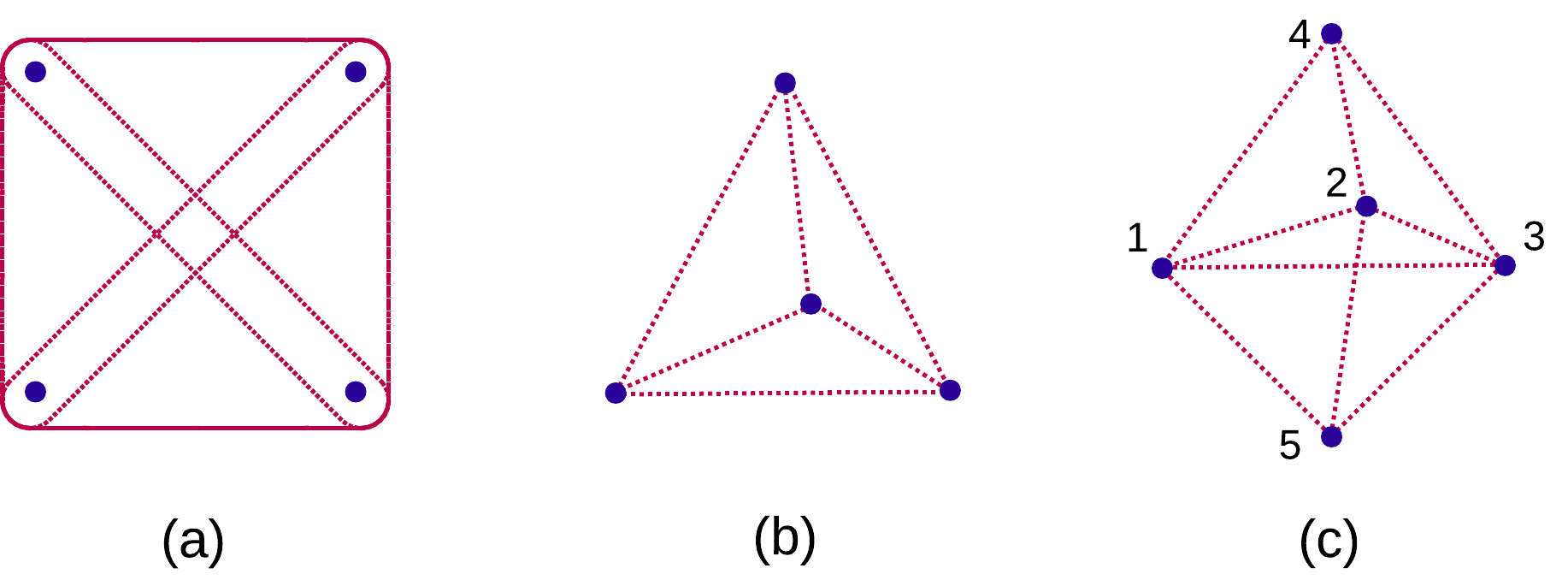}
\end{center}
\caption{(a, b) The graph of the four-qubit hypergraph state No. 14 can 
also be depicted as a tetrahedron, where the triangles on the surface represent
the three-edges. (c) The hypergraph of the only three-uniform five-qubit state
with maximally mixed reduced states can be represented in a similar way. This hypergraph has 7 edges corresponding to the triangles. In fact, all possible 
three-edges apart from the ones that contain $4$ and $5$ are present.
}
\label{hg-bild-fuenf}
\end{figure}

For five qubits, one directly finds that there is exactly one further
$k$-uniform hypergraph state with this property. The hypergraph 
is three-uniform and  depicted in Fig.~\ref{hg-bild-fuenf}. The  hypergraph state is 
symmetric on the first three qubits, and, after applying Hadamard 
transformations on all qubits, it can be written as
\begin{align}
\ket{F_{1}} &
= \frac{1}{\sqrt{2}}\ket{GHZ_5} 
+ \frac{1}{2} \ket{\eta} \ket{00}
+ \frac{1}{2} \ket{\overline{\eta}} \ket{11},
\end{align}
where $\ket{GHZ_5}=(\ket{00000}+ \ket{11111})/\sqrt{2}$ is the five-qubit
GHZ state, $\ket{\eta}=(\ket{100}+\ket{010}+\ket{001}-\ket{111})/2$
is a three-qubit state unitarily equivalent to the three-qubit GHZ state, 
and $\ket{\overline{\eta}}=X_1 X_2 X_3\ket{\eta}$ arises from $\ket{\eta}$ when
all qubits are flipped.

The geometric measure of entanglement for this state is $E_G=0.6000$,
and the genuine negativity is $N_G = 1/2$. The overlap with biseparable 
states is $\alpha_S = (3+\sqrt{5})/8\approx 0.65450$, and the noise 
robustness is given by $p_{\rm ppt} = 0.6070.$


\subsection{Six qubits}
Also for six qubits, a full characterization of all $k$-uniform hypergraph
states with maximally mixed reduced single-particle states is possible. 
Apart from the usual graph states, one finds such states only for 
three-uniform hypergraphs. It turns our that there are 24 different 
classes, a detailed list is given in Table \ref{6qbtable} in the Appendix. 
For these states one can say the following: All of them are LU inequivalent, 
as for all of them the geometric measure $E_G$ is different. For none of 
these states, also all of the two-qubit reduced density matrices are 
maximally mixed.

Two states in this set may be of interest: First the state which has 
(after appropriate local unitary transformations) the simplest form in
the standard basis is the representative of the first class, 
\begin{align}
\ket{S_1}&=\frac{1}{8}\big[
(\ket{000000}-\ket{111111})
+ (\ket{001}+\ket{010}+\ket{100})\otimes \ket{000}+
\nonumber
\\
&
+ (\ket{110}+\ket{101}+\ket{011})\otimes \ket{111}
\big].
\end{align}
This state has a geometric measure $E_G(\ket{S_1})=0.6.$
The state with the highest geometric measure of entanglement is the 
representative of the 21th class, and its measure is $E_G(\ket{S_{21}})=0.817.$

\section{General conditions on local unitary equivalence 
and local Pauli equivalence \label{sec:LU_LMES}}

As shown in the previous section, the finite set of local Pauli operations 
is sufficient to characterize the LU equivalence classes for hypergraph 
states of four qubits. So the question arises, under which conditions this
is true in more general cases. We note that it cannot be true in general, 
since for the usual graph states not even the larger set of local Clifford
operations is sufficient to characterize LU equivalence.

In this section we establish general conditions, under which local Pauli
(LP) transformations are the only transformations that have to be checked
in order to prove LU equivalence. We first review a necessary and sufficient 
condition for LU equivalence of general $N$-partite states presented in 
Refs.~\cite{Kraus2010,Kraus2010a}, and then show that for many hypergraph 
states this can be reduced to equivalence under the action of LP operators, 
which transform the corresponding hypergraph in a simple way.

\subsection{Local unitary equivalence of multipartite states}

First, we will use the notation $\varrho_i = \ptr{\neg i}(\ketbra{\psi})$ for
the single-qubit reduced state of $\ket{\psi}$, and equivalently $\sigma_i$
for $\ket{\phi}$. It has been shown that two generic states, where the single-qubit
reduced states are not maximally mixed (that is, 
$\varrho_i,\sigma_i\not\propto\mathbb{1}$ for all $i$), are LU equivalent 
if and only if their corresponding so-called {\it unique standard form} 
coincides.  Moreover, for non-generic states, there exists an algorithm 
that determines the local unitaries that transform $\ket{\phi}$ to $\ket{\psi}$ 
in case they exist \cite{Kraus2010,Kraus2010a}.

As we are concerned here with the LU equivalence of generic states, 
we briefly recall the criteria of LU equivalence in this case. The 
standard form $\ket{\psi^\mathrm{s}}$ of an $N$-partite state $\ket{\psi}$ 
can be obtained in three steps~\cite{Kraus2010}. First, the state is 
transformed to its \emph{trace decomposition},
\begin{equation}
  \ket{\psi^\mathrm{t}}
  =
  U^{\rm t}_1 \otimes \cdots \otimes U^{\rm t}_N \ket{\psi},
\end{equation}
where $U^{\rm t}_i$ are local unitaries such that the single-qubit reduced
states are diagonal, i.e.\ $\varrho^{\rm t}_i=U^{\rm t}_i\varrho_i(U^{\rm t}_i)^\dag =
D_i \equiv \diag (\mu_i^1,\mu_i^2)$ in the standard basis. Second, 
the \emph{sorted trace decomposition} $\ket{\psi^\mathrm{st}}$ is the trace 
decomposition with $\mu_i^1\ge\mu_i^2$ for all $i$. This decomposition 
can be reached by applying local $X$ operators to $\ket{\psi^\mathrm{t}}$ 
when necessary. The sorted trace decomposition is already unique, up to a 
global phase $\alpha_0$ and local phase operators 
$Z_i(\alpha_i)=\diag(1,\e^{i\alpha_i})$. That is,
$\e^{i\alpha_0}\bigotimes_iZ_i(\alpha_i)\ket{\psi^\mathrm{st}}$ 
is also a sorted trace decomposition of $\ket{\psi}$. So finally, 
the sorted trace decomposition is made unique by imposing the following 
conditions on the phases. We write the sorted trace decomposition in 
the computational basis,
$\ket{\psi^\mathrm{st}}=\sum_{\bm{i}} \lambda_{\bm{i}} \ket{\bm{i}}$, with
$\bm{i}\in\mathbb{Z}_2^N$ being a binary vector of length $N$. Then, from the set
$\Lambda=\{\bm{i}:\lambda_{\bm{i}}\neq0\}$, we pick a maximal set of linearly 
independent vectors $\bm{i}$ by going through the set $\Lambda$ in the order 
given by the computational basis. These vectors then form the set 
$\bar{\Lambda}\subseteq \Lambda$. The global phase $\alpha_0$ is set such that
\begin{align}
  \lambda_{\bm{0}} > 0    &\quad \text{if $\lambda_{\bm{0}} \neq 0$} \\
  \lambda_{\bm{i_0}} > 0  &\quad \text{if $\lambda_{\bm{0}} = 0$},
\end{align}
where $\bm{i_0}$ is the first linearly \emph{dependent} vector in $\Lambda$.
The other $N$ phases $\alpha_1,\dots,\alpha_N$ are set so that $\e^{i\alpha_0}\lambda_{\bm{i}}>0$ for $\bm{i}\in\bar{\Lambda}$.
Since the resulting \emph{standard form} $\ket{\psi^\mathrm{s}}$ is unique, we have
\begin{equation}
  \ket{\psi} \lueq \ket{\phi}
  \textrm{ iff }
  \ket{\psi^\mathrm{s}} = \ket{\phi^\mathrm{s}}
  .
\end{equation}
This necessary and sufficient condition can be rewritten in the following
way \cite{Kraus2010a}:
\begin{align}
  \ket{\psi} \lueq \ket{\phi} &
  \textrm{ if and only if:}
  \label{eq:LUequiv}
  \\
  \mbox{ there exist } & \{\alpha_i\}_{i=0}^N,\bm{k},\{U_i^\mathrm{t}\}_{i=1}^N,
  \{V_i^\mathrm{t}\}_{i=1}^N
  \mbox{ such that }
  \bigotimes_i U^\mathrm{t}_i \ket{\psi} = 
  \e^{i\alpha_0} \bigotimes_i Z_i(\alpha_i) X_i^{k_i} V^\mathrm{t}_i \ket{\phi}
  .
  \nonumber
\end{align}
Here $U^\mathrm{t}_i$ and $V^\mathrm{t}_i$ are the local unitaries that
transform $\ket{\psi}$ and $\ket{\phi}$ to their trace decompositions
$\ket{\psi^\mathrm{t}}$ and $\ket{\phi^\mathrm{t}}$, respectively, and
$\bm{k}\in\mathbb{Z}_2^N$ is a bit string such that the operators $X_i^{k_i}$
make the order of the eigenvalues of the single qubit reduced states
$U_i^\mathrm{t}\varrho_i(U_i^\mathrm{t})^\dag$ and
$V_i^\mathrm{t}\sigma_i(V_i^\mathrm{t})^\dag$ coincide. Then, one only needs
to check that a solution for the set $\{\alpha_i\}_{i=0}^N$ exists. 
In this case, the unitaries of
Eq.~(\ref{eq:LUequiv_def}) that transform $\ket{\phi}$ to $\ket{\psi}$ (up to
a global phase) are $U_i=(U_i^\mathrm{t})^\dag
Z_i(\alpha_i)X_i^{k_i}V_i^\mathrm{t}$.

\subsection{LU equivalence of generic LME states}

As shown in Ref.~\cite{Kruszynska2009}, a state is an LME state 
if and only if it is LU equivalent to
\begin{equation}
  \ket{\psi} = \frac{1}{\sqrt{2^N}} \sum_{\bm{x}\in\mathbb{Z}_2^N}
  \e^{i\alpha_{\bm{x}}} \ket{\bm{x}}
  ,
  \label{eq:LMES_Zbasis}
\end{equation}
with $\alpha_{\bm{x}}$ being real phases. Such a general LME state with 
arbitrary phases is brought to the trace decomposition by local unitaries 
$HZ(-\beta_i)$, where
$\cot(\beta_i)=\frac{\avg{X_i}}{\avg{Y_i}}$ if $\avg{Y_i}\neq0$ and
$\beta_i=0$ else~\cite{Kraus2010a}. We call the state
$\bigotimes_iZ_i(-\beta_i)\ket{\psi}$ an LME state in the $Z$ basis. In the
following, we consider an LME states to be always in this basis, unless otherwise
stated. This means that for every qubit $i$, the single-qubit reduced state is
of the form
\begin{equation}
  \varrho_i\propto\begin{pmatrix}1&x\\x&1\end{pmatrix}
\end{equation}
with $x\in\mathbb{R}$. If the state is generic, then $x\neq0$ since the reduced
states are not maximally mixed. In this basis,
the state is brought to the trace decomposition by applying the Hadamard 
transformation $H^{\otimes N}$.

Hence, for $\ket{\psi}$ and $\ket{\phi}$ generic LME states in the $Z$ basis,
the LU equivalence condition (\ref{eq:LUequiv}) reads
\begin{align}
  \ket{\psi} \lueq \ket{\phi} &
  \textrm{ if and only if:}
  \nonumber
  \\
  \mbox{ there exist } 
  &
  \{\alpha_i\}_{i=0}^N,\bm{k}
  \mbox{ such that }
  \bigotimes_i H_i \ket{\psi} = \e^{i\alpha_0} \bigotimes_i Z_i(\alpha_i) X_i^{k_i} H_i \ket{\phi}
  ,
\end{align}
and the unitaries (\ref{eq:LUequiv_def}) that transform $\ket{\phi}$ to
$\ket{\psi}$ must be of the form
\begin{equation}
  U_i
  =
  H_i Z_i(\alpha_i) X_i^{k_i} H_i
  =
  X_i(\alpha_i) Z_i^{k_i}
  ,
 \label{eq:local_unitaries_c11}
\end{equation}
up to a global phase.
Here we defined $X_i(\alpha_i)\equiv H_iZ_i(\alpha_i)H_i$.
As one can see, the allowed unitaries are already quite restricted, as all the
unitaries bringing a generic LME state to its trace-decomposition are the same.
That is, the difference from just Pauli operators comes from the freedom of
local phase operators in the sorted trace decomposition.

In the following lemma we show that, for certain LME states, this freedom can
be further reduced and it is enough to consider $\alpha_i\in\{0,\pi\}$, i.e.\
either the identity or $X_i$, which implies that in this basis, and up to a
global phase, $\ket{\psi}$ is LU equivalent to $\ket{\phi}$ if and only if
they are local Pauli (LP) equivalent. First, we define
$M_{\ket{\phi^\mathrm{t}}}$ as the matrix whose rows are the first linearly
independent vectors $\bm{x}$ such that
$\braket{\bm{x}}{\phi^\mathrm{t}}\neq0$. That is, $M_{\ket{\phi^\mathrm{t}}}$
is analogous to the set $\bar{\Lambda}$ but for the trace decomposition, and
expressed in matrix form.

\begin{lem}
  \label{lem:localphase_localZ}
  Let $\ket{\psi^\mathrm{t}}$ and $\ket{\phi^\mathrm{t}}$ be $N$-qubit LME
  states in the trace-decomposition that fulfill the following conditions:
  \begin{enumerate}
    \item They have real coefficients
      \label{lem:localphase_localZ_real}
      $\braket{\bm{x}}{\psi^\mathrm{t}},\braket{\bm{x}}{\phi^\mathrm{t}}\in\mathbb{R}$.
    \item Coefficients $\braket{\bm{0}}{\psi^\mathrm{t}}$ and
      $\braket{\bm{0}}{\phi^\mathrm{t}}$ are different from 0.
      \label{lem:localphase_localZ_0}
    \item $M_{\ket{\phi^\mathrm{t}}}$ is square and invertible in
      $\mathbb{Z}_2$, i.e.| $|M_{\ket{\phi^\mathrm{t}}}|\equiv1\pmod{2}$.
      \label{lem:localphase_localZ_matrix}
  \end{enumerate}
  Then, $\ket{\phi^\mathrm{t}}$ can be converted to $\ket{\psi^\mathrm{t}}$ by
  local unitary phase operators if and only if $\ket{\phi^\mathrm{t}}$ can be
  converted to $\ket{\psi^\mathrm{t}}$ by local $Z$ operators, i.e.\
  \begin{align}
  &
    \mbox{ There exist } \{\alpha_i\}_{i=0}^n \in\mathbb{R} 
    \mbox{ such that }
    \ket{\psi^\mathrm{t}} = \e^{i\alpha_0} \bigotimes_{i=1}^nZ_i(\alpha_i) 
    \ket{\phi^\mathrm{t}} 
    \textrm{ if and only if}
    \nonumber
    \\
    &
     \mbox{ There exist } \tilde{\alpha}_0\in\mathbb{R}, \mbox{ and }\bm{l}\in\mathbb{Z}_2^n 
     \mbox{ such that }
    \ket{\psi^\mathrm{t}} = \e^{i\tilde{\alpha}_0} \bigotimes_{i=1}^nZ_i^{l_i} \ket{\phi^\mathrm{t}}
    .
  \end{align}
\end{lem}
Note that the lemma considers states $\ket{\psi^\mathrm{t}}$ and
$\ket{\phi^\mathrm{t}}$ that are related by local phase operators, so their
coefficients have the same absolute value, i.e.\
$|\braket{\bm{x}}{\psi^\mathrm{t}}|=|\braket{\bm{x}}{\phi^\mathrm{t}}|$ for
all $\bm{x}$. Moreover, since the coefficients are real they can only differ
by a sign.

\begin{proof}
  The sufficient condition is straightforward, as one can set $\alpha_i=\pi
  l_i$ and $\alpha_0=\tilde{\alpha}_0$. For the necessary condition, note that
  due to condition (\ref{lem:localphase_localZ_real}) we have that
  $\ket{\psi^*}=\ket{\psi}$ and $\ket{\phi^*}=\ket{\phi}$, where
  $\ket{\psi^*}$ denotes the complex conjugate of $\ket{\psi}$ in the
  computational basis. Hence, we can impose the following condition on the
  phases:
  \begin{equation}
    \ket{\phi^\mathrm{t}} = \e^{i2\alpha_0} \bigotimes_i Z_i(2\alpha_i) \ket{\phi^\mathrm{t}}
    \label{eq:cc_condition1}
    .
  \end{equation}
  Since $\braket{\bm{0}}{\phi^\mathrm{t}}\neq0$, we can see that $\alpha_0$
  is fixed to either $0$ or $\pi$. Since the coefficients of 
  $\ket{\phi^\mathrm{t}}$ and $\ket{\psi^\mathrm{t}}$ can differ only by a sign, 
  we have the more general condition
  \begin{equation}
    \alpha_0 + \bm{x}^\intercal\bm{\alpha} = \pi k_{\bm{x}} \ \forall \bm{x}\in
    S_{\phi^\mathrm{t}}
    \label{eq:cc_condition1_phases}
    ,
  \end{equation}
  where we defined $\bm{\alpha}=(\alpha_1,\dots,\alpha_n)^\intercal\in\mathbb{R}^n$,
  and $S_{\phi^\mathrm{t}}=\{\bm{x}:\braket{\bm{x}}{\phi^\mathrm{t}}\neq0\}$.
  Note that $S_{\phi^\mathrm{t}}$ is analogous to the set $\Lambda$ previously
  defined, but for the trace decomposition.
  Using the matrix $M=M_{\ket{\phi^\mathrm{t}}}$ and that
  $\alpha_0\in\{0,\pi\}$,
  Eq.~(\ref{eq:cc_condition1_phases}) can also be rewritten as
  \begin{equation}
    M\bm{\alpha} = \pi\bm{k}.
    \label{eq:cc_condition1_short}
  \end{equation}
  Here $\bm{k}\in\mathbb{Z}_2^N$ is the vector whose entries are
  $k_{\bm{x}}+\frac{\alpha_0}{\pi} \pmod{2}$, with $\bm{x}$ the vectors
  forming the matrix $M$.

  As $M$ is an integer matrix, so is the adjugate matrix
  $A\equiv|M| M^{-1}$ (that is, the transpose of the cofactor matrix,
  $A=C^\intercal$, where each entry $C_{ij}$ is the $i,j$-minor of $M$). Using
  this, Eq.~(\ref{eq:cc_condition1_short}) can be written as
  \begin{equation}
    \frac{|M| \bm{\alpha}}{\pi}
    =
    A\bm{k}
    \pmod{2}
    .
  \end{equation}
  As $M$ is invertible in $\mathbb{Z}_2$, $|M|\equiv1\pmod{2}$, so
  the vector $\bm{\alpha}=\pi\bm{l}$, with
  $\bm{l}\equiv A\bm{k}\pmod{2}\in\mathbb{Z}_2^N$, is a
  solution to the previous equation. This shows that the phases can always be
  chosen to be either $0$ or $\pi$, and hence $\ket{\psi^\mathrm{t}}$ and
  $\ket{\phi^\mathrm{t}}$ are equivalent under the action of $Z$ operators.
\end{proof}
Note that the conditions
$\braket{\bm{0}}{\psi^\mathrm{t}}\neq0$
and
$\braket{\bm{0}}{\phi^\mathrm{t}}\neq0$
can
be relaxed, as long as there exists at least a linearly \emph{dependent}
vector $\bm{x}$ such that the global phase can be fixed to $0$ or $\pi$.

We are now going to use this lemma to derive a simple necessary and sufficient
condition for LU equivalence of certain LME states. To this end, we redefine
generic LME states by adding a second condition that excludes a set of zero
measure. We call an LME state $\ket{\phi}$ generic if it fulfills the
  following two conditions:
(i) none of the single-qubit reduced states are completely mixed (i.e.\
      $\varrho_i\not\propto\mathbb{1}\ \forall i$); and
(ii) let $\ket{\phi}$ be in the $Z$ basis, then there exists a
      $\bm{k}\in\mathbb{Z}_2^n$ such that the state
      $\ket{\phi_{\bm{k}}}=Z^{\bm{k}}\ket{\phi}$ fulfills that
      $\bra{\bm{0}}H^{\otimes n}\ket{\phi_{\bm{k}}}\neq0$ and
      $\bra{\bm{x}}H^{\otimes n}\ket{\phi_{\bm{k}}}\neq0$ for all
      $\bm{x}\in\mathbb{Z}_2^n$ with $|\bm{x}|=1$.
      ($\ket{\bm{x}}$ is a state of the computational basis.)
Nearly all LME states fulfill these two conditions. Note also that for our
purposes condition (ii), which is equivalent to
$M_{H^{\otimes n}\ket{\phi_{\bm{k}}}}=\mathbb{1}$, can be relaxed to 
$|M_{H^{\otimes n}\ket{\phi_{\bm{k}}}}|\equiv1\pmod{2}$.

As we mentioned before, the allowed unitaries that transform $\ket{\phi}$ to
$\ket{\psi}$ can be different from Pauli operators due to the freedom of local
phase operators in the sorted trace decomposition. A generic LME state $\ket{\phi}$
always has a Pauli equivalent state $\ket{\phi_{\bm{k}}}$ that fulfills
conditions (\ref{lem:localphase_localZ_0}) and
(\ref{lem:localphase_localZ_matrix}) of Lemma~\ref{lem:localphase_localZ}.
This leads to the proof of the following theorem:
\begin{thm}
  \label{th:lu_pauli_equiv}
  Let $\ket{\psi}$ and $\ket{\phi}$ be $N$-qubit generic LME states in the $Z$
  basis with real coefficients in the trace decomposition.
  Then $\ket{\psi}\lueq\ket{\phi}$ if and only if $\ket{\psi}$ and
  $\ket{\phi}$ are equivalent under the action of local Pauli operators (in
  the $Z$ basis).
\end{thm}

\begin{proof}
  The sufficient condition is straightforward.
  For the necessary condition, note that $\ket{\phi}\equiv\ket{\phi_{\bm{0}}}$
  is equivalent under Pauli operators to $\ket{\phi_{\bm{k}}}$ for all $\bm{k}$.
  Hence, proving the theorem for any $\ket{\psi_{\bm{k}^\prime}}$ and any
  $\ket{\phi_{\bm{k}}}$ is enough. Recall that
  $\ket{\psi}\lueq\ket{\phi_{\bm{k}}}$ if and only if
   there exists a binary vector $\bm{k}^\prime$ and
  real phases $\{\alpha_i\}_{i=0}^N$ such that 
  (see also Eq.~(\ref{eq:local_unitaries_c11}))
  \begin{equation}
    \ket{\psi_{\bm{k}^\prime}}
    =
    \bigotimes_{i=1}^N Z_i^{k_i^\prime} \ket{\psi}
    =
    \e^{i\alpha_0} \bigotimes_{i=1}^N X_i(\alpha_i) \ket{\phi_{\bm{k}}}
    ,
  \end{equation}
  or equivalently if there exists such vector $\bm{k}^\prime$ and phases
  $\{\alpha_i\}_{i=0}^N$, and
  \begin{equation}
    \ket{\psi_{\bm{k}^\prime}^\mathrm{t}}
    =
    \e^{i\alpha_0} \bigotimes_{i=1}^N Z_i(\alpha_i) \ket{\phi_{\bm{k}}^\mathrm{t}}
    .
  \end{equation}
  We choose $\ket{\phi_{\bm{k}}}$ such that
  $\braket{\bm{0}}{\phi_{\bm{k}}^\mathrm{t}}\neq0$ and
  $M_{\ket{\phi_{\bm{k}}^\mathrm{t}}}=\mathbb{1}$, which is always possible
  due to condition (ii) of generic states. Then, due to
  Lemma~\ref{lem:localphase_localZ} $\ket{\psi_{\bm{k}^\prime}^\mathrm{t}}$
  and $\ket{\phi_{\bm{k}}^\mathrm{t}}$ are equivalent under local $Z$
  operators, i.e.\ $\ket{\psi_{\bm{k}^\prime}}$ and $\ket{\phi_{\bm{k}}}$ are
  equivalent under local $X$ operators. Each of these states are equivalent to
  $\ket{\psi}$ and $\ket{\phi}$ under the actions of local $Z$ operators,
  respectively, which concludes the proof.
\end{proof}

\subsection{LU equivalence of hypergraph states}

Now we are able to show that in many cases, two hypergraph states are 
LU equivalent if and only if they are LP equivalent. That is, if and 
only if their associated hypergraphs are related by a sequence of the 
hypergraph transformations studied in Section~\ref{sec:graph_trafos}. 
We identify some conditions that have to be fulfilled by two $N$-qubit
hypergraph states so that they are LU equivalent if and only if they are LP
equivalent.

A hypergraph state in the computational basis has only real coefficients,
\begin{equation}
  \ket{\psi} = \frac{1}{\sqrt{2^N}} \sum_{\bm{x}\in\mathbb{Z}_2^N} (-1)^{f(\bm{x})} \ket{\bm{x}}
  ,
  \label{eq:hgs_Zbasis}
\end{equation}
where
\begin{equation}
  f(\bm{x})=\bigoplus_{e\in E}\prod_{i\in e}x_i
  \label{eq:booleanf}
\end{equation}
is a boolean function expressed as a sum of polynomials (modulo 2) that
depends on the hypergraph structure (and we define $\prod_{i\in\emptyset}x_i\equiv1$ 
for the empty edge). The degree $d$ of such a polynomial
representation is defined as $d=\max\{|e|:e\in E\}$, where $|e|$ stands for
the cardinality of the edge $e$. Hence, $d$ equals the cardinality of the
largest edge that appears in the graph.

If $\ket{\psi}$ is a generic LME state, 
then the unitaries bringing it to the trace decomposition are $H^{\otimes N}$. 
This leads to the following corollary:
\begin{cor}
  Let $\ket{\psi}$ and $\ket{\phi}$ be $N$-qubit hypergraph states that 
  belong to the family of generic LME states [see conditions (i) and (ii) 
  of the definition]. Then $\ket{\psi}\lueq\ket{\phi}$ if and only if 
  $\ket{\psi}$ and $\ket{\phi}$ are equivalent under the action of local 
  Pauli operators.
\end{cor}

Note that by far not all hypergraph states fulfill the condition
stated in the corollary. In fact, for the four-qubit hypergraphs 
in Fig.~(\ref{hg-bild-vier}) only 14 of 27 
fulfill the condition that none of the reduced density matrices is
maximally mixed. Further notable exceptions to the first condition 
of generic states are the connected graph states (with 
$\varrho_i\propto\mathbb{1}$  for all $i$) and other $k$-uniform 
hypergraph states, such as the states from Table~\ref{6qbtable}. 
Hence, the set of connected graph states, for which in general 
local Clifford operations are not sufficient to characterize LU 
equivalence, is excluded. The second condition allows us to ignore 
some problematic states, where the Lemma 3 cannot be applied.

We have explained above that two $k$- and
$k^\prime$-uniform hypergraph states are not 
equivalent under the action of local Pauli 
operators unless they are the same.  
So we obtain the following corollary:

\begin{cor}
  \label{cor:uniform_not_LU}
  Let $\ket{\psi}$ and $\ket{\phi}$ be two different uniform $N$-qubit 
  hypergraph states that  belong to the family of generic LME states.
  Then $\ket{\psi}$ is not LU equivalent to $\ket{\phi}$.
\end{cor}

Finally, we are able to identify a large class of states where Corollary 5
can be applied.
More precisely, we show that if two hypergraphs contain 
the maximal edge $\{1,\dots,N\}$, then the two
corresponding hypergraph states are LU equivalent if and only if they are LP
equivalent. In order to do so, consider again the boolean function $f(\bm{x})$ defined in
Eq.~(\ref{eq:booleanf}). The quantity $(-1)^{f(\bm{x})}$, that defines the
amplitudes of the hypergraph state, can be conveniently rewritten as
\begin{equation}
  (-1)^{f(\bm{x})} = 1-2f(\bm{x}) \equiv f_\pm(\bm{x})
  .
  \label{eq:coeff_f}
\end{equation}
The Hadamard transform of $f(\bm{x})$ is defined as
\begin{equation}
  \hat f (\bm{w})=\frac{1}{2^N} \sum_x f(\bm{x})(-1)^{\bm{x}^\intercal\bm{w}}.
\end{equation}
The Hadamard transform of $f_\pm$, denoted by $\hat{f}_\pm$, gives the
coefficients of the hypergraph state in the trace decomposition, i.e.
\begin{align}
  \hat{f}_\pm(\bm{w})
  &= \frac{1}{2^N} \sum_{\bm{x}} (-1)^{f(\bm{x})}(-1)^{\bm{x}^\intercal\bm{w}}
  \nonumber\\
  &= \frac{1}{2^N} \sum_{\bm{x}} \left[ 1-2f(\bm{x}) \right](-1)^{\bm{x}^\intercal\bm{w}}
  \nonumber\\
  &= \delta_{\bm{w},\bm{0}}-2\hat{f}(\bm{w})
  \label{ht}
  ,
\end{align}
since the linear function $\bm{x}^\intercal\bm{w}$ of $\bm{x}$ is always
balanced as long as $\bm{w}\neq\bm{0}$.

The cardinality of the support of $\hat{f}$, i.e.\ the number of strings
$\bm{w}$ such that $\hat{f}(\bm{w})\neq0$, is bounded from below as
$|\supp(\hat{f})|\geq 2^{d}$ \cite{degree}. Thus, using Eq.~(\ref{ht}), it
follows that $|\supp(\hat f_\pm)|\geq 2^{d}-1$, as $\hat f_\pm(\bm{w})$ is
proportional to $\hat f(\bm{w})$ whenever $\bm{w}\neq\bm{0}$.
If the hypergraph contains the largest edge of cardinality $N$, then the
associated boolean function $f$ has degree $d=N$. Furthermore, the cardinality
of $\supp(f)$, that we denote by $F\equiv|\supp(f)|$, is always an odd
number~\cite{Qu2013_entropic}. This can easily seen by noting that every $C_e$
different from the one acting on all qubits contains an even number of $-1$,
and hence also the product of different $C_e$ does. Multiplying this product
with the operator $C_V$ on all qubits leads then to an operator with an odd
number of $-1$. As a consequence, the function $f$ is never balanced. It is
also easy to see that the off-diagonal term of the single-qubit reduced state
$\varrho_i$ is \cite{Qu2013_entropic}
\begin{align}
  (\varrho_i)_{01}
  &= \frac{1}{2^N} \sum_{\bm{x}\in\mathbb{Z}_2^{N_i}} (-1)^{f^{(i)}(\bm{x})}
  \nonumber\\
  &= \frac{1}{2^N} \left( 2^{N-1}-2F^{(i)} \right)
  ,
\end{align}
where we defined $N_i$ as the number of neighbours of vertex $i$, $f^{(i)}$ as
the boolean function $f^{(i)}=\bigoplus_{e\in E^{(i)}}\prod_{j\in e}x_j$, and
$F^{(i)}$ as the cardinality of its support, i.e.\ $F^{(i)}=|\supp(f^{(i)})|$. By
the same reasoning as before, $f^{(i)}$ corresponds again to a hypergraph
which contains the maximal edge (here of cardinality $N-1$), and is thus also
not balanced and $F^{(i)}\neq2^{N-2}$ for any $i$.
Hence, none of the single-qubit reduced states is proportional to the identity.
Moreover, $f$ not balanced also implies that $\hat f_\pm (\bm{0})\neq 0$, and
therefore $|\supp(\hat f_\pm)|=2^{N}$, that is $\hat f_\pm$ is never vanishing. This means
that a hypergraph state with the edge of maximal cardinality is a generic LME
state, which leads to the corollary:
\begin{cor}
  Let $\ket{\psi}$ and $\ket{\phi}$ be $N$-qubit hypergraph states such that
  the associated hypergraphs contain the edge of maximum cardinality $N$.
  Then $\ket{\psi}$ and $\ket{\phi}$ are generic LME states and
  $\ket{\psi}\lueq\ket{\phi}$ if and only if $\ket{\psi}$ and $\ket{\phi}$ are
  equivalent under the action of local Pauli operators (in the $Z$ basis).
\end{cor}

\section{Nonclassicality of hypergraph states}

In this section, we will investigate whether the correlations present 
in hypergraph states can be used to demonstrate contradictions between
quantum physics and the classical world view. Classical models obey 
by definition the constraint of realism: Any observable has a given
value and this value is independent of whether the observable is 
indeed measured or not. In order to arrive at a contradiction to quantum
mechanics, further assumptions on a classical model are needed. One 
possible assumption is noncontextuality. This means that the value
of one measurement is independent of which other compatible measurement
is carried out jointly with it. To give an example, one may consider the
situation where the measurements $A$ and $B$ commute, and $A$ and $C$
commute. Then, the outcome of $A$ is assumed to be independent  of 
whether it is measured together with $B$ or with $C$.
A different type of 
assumption is the locality assumption. Here, one considers a system 
consisting of several particles and demands that an observable on one 
particle has a value that is independent of other observables at distant 
particles, measured at the same time. Clearly, the locality assumption can 
be viewed as the noncontextuality assumption applied to special observables. 

Both assumptions are in conflict with the quantum mechanical prediction. 
For noncontextuality, this statement is known as the Kochen-Specker theorem 
\cite{Specker60, KS67}, while for locality this fact is called Bell's theorem \cite{Bell64}. It is 
also known that for usual graph states the stabilizer formalism can be used
to derive arguments and inequalities to demonstrate these contradictions. For
locality, the most famous of the arguments is the GHZ reasoning \cite{GHZargument},
but a plethora of further results exists \cite{merminineq,S04,GTHB04}. For noncontextuality, the most known argument
is the Mermin star \cite{Mermin93}, but recently many other contradictions have been 
systematically investigated \cite{Cabello08, morde}. We will show now that similar contradictions
can be found employing the nonlocal stabilizer of hypergraphs. We will first derive 
an argument and an inequality for noncontextuality, then we will show that 
it can be rephrased as a Bell inequality.

\subsection{A GHZ-like argument for the hypergraph formalism}
To start, let us recall the GHZ argument and the Mermin inequality
for the three-qubit GHZ state. The three-qubit GHZ state is a graph 
state, where the graph is the fully connected graph. Its stabilizing
operators are therefore given by $g_1 = X_1 Z_2 Z_3$, $g_2 = Z_1 X_2 Z_3$, 
and $g_3 = Z_1 Z_2 X_3$. Then one considers the operator
\begin{align}
\MM &= g_1 + g_2 + g_3 + g_1 g_2 g_3
\label{mermin1}
\\
& = X_1 Z_2 Z_3 + Z_1 X_2 Z_3 + Z_1 Z_2 X_3 - X_1 X_2 X_3.
\label{mermin2}
\end{align}
The GHZ state is an eigenstate with the eigenvalue $+1$ for all 
the operators in the sum in Eq.~(\ref{mermin1}), so the expectation value for this state
is $\mean{\MM}=4.$ On the other hand, if one considers the $X_i$ 
and $Z_i$ in the second line as classical quantities with values 
$\pm 1$, the product of the first three terms in Eq.~(\ref{mermin2}) 
equals the last term, up to the sign. Therefore it is impossible to 
assign values to the $X_i$ and $Z_i$ such that all four terms take 
the value $+1$. It follows that classical models have to obey the 
constraint $\mean{\MM}\leq 2$, which is known as Mermin inequality. 

The question arises, whether a similar argument can be derived 
from the hypergraph formalism. Since the hypergraph stabilizer 
consists of nonlocal operators, such an argument will, in the first
place, lead to a noncontextuality inequality. Later we will discuss
whether this can be translated to a Bell inequality.

To derive an inequality in the hypergraph formalism, we aim at 
finding a hypergraph state whose stabilizer operators fulfill the following relation:
\be
\sum_{i=1}^{N} g_i  + \prod_{i=1}^N g_i = 
\sum_{i=1}^{N} X_i \big(\prod_{e \in E(i)} C_{e\setminus\{i\}}\big)
- \prod_{i=1}^N X_i
.
\label{ansatzghz}
\ee
In this case, any classical assignment of $\pm 1$ to the $X_i$ as well 
as to the operators $C_{e\setminus\{i\}}$ (which commute) cannot reproduce the quantum mechanical 
predictions. In order to make this ansatz work, we consider the subset 
of $k$-uniform fully connected hypergraph states. These are the hypergraphs 
where any possible $k$-edge is present. An example is the graph No.~14 in 
Fig.~\ref{hg-bild-vier}, which is the fully connected three-uniform 
four-vertex graph.

\begin{figure}[t]
\begin{center}
\includegraphics[width=0.9\textwidth]{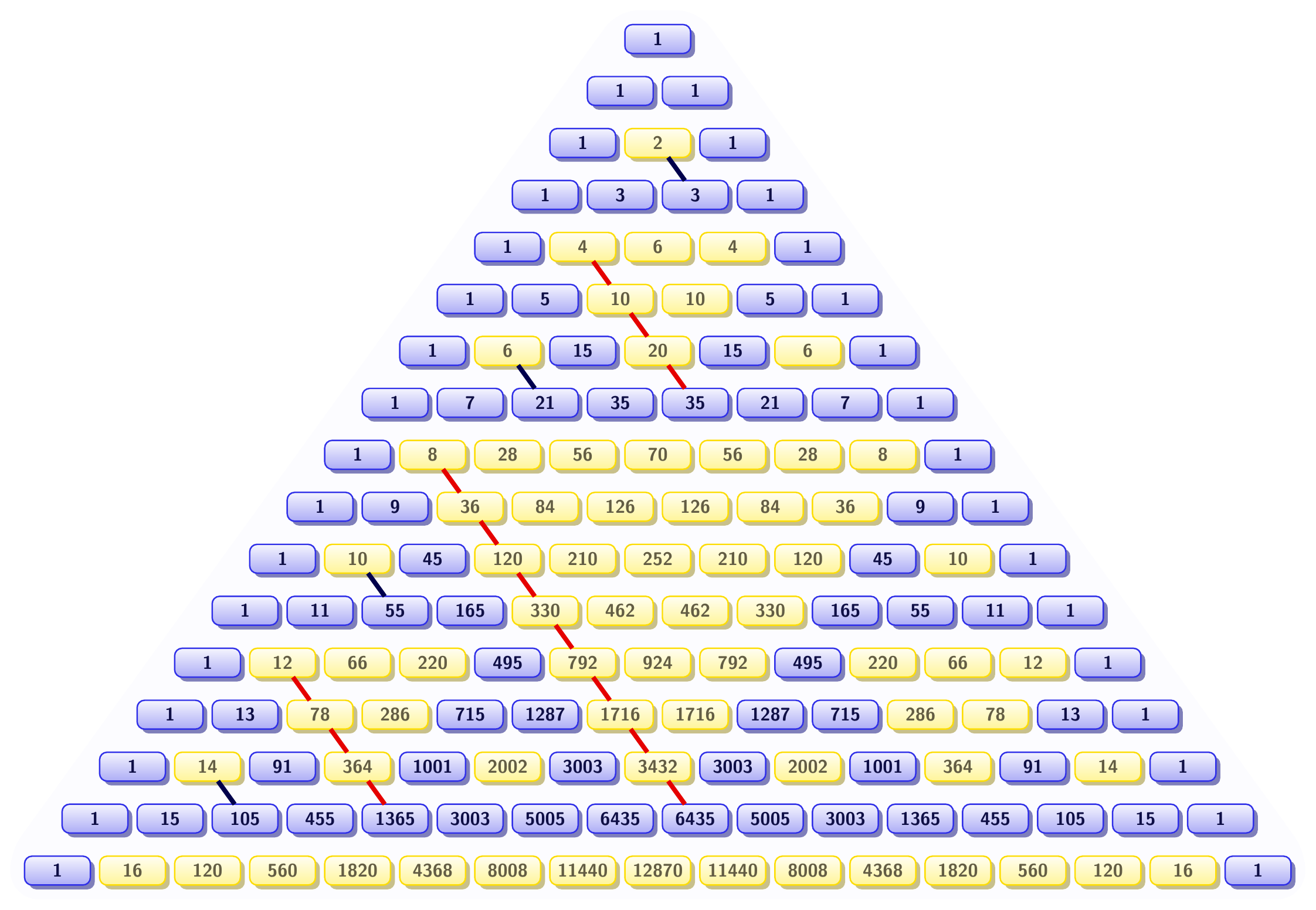}
\end{center}
\caption{Pascal's triangle allows for a visualization of the
conditions for obtaining a noncontextuality inequality in
the hypergraph formalism. For that, one has to find a tuple
$(N,k)$ where ${N \choose k}$ is odd, while ${N-\alpha \choose k-\alpha}$ 
is even for any $\alpha < k$. Obvious solutions are of the type 
$(N,k)=(3+4r,2)$, denoted by black lines. These correspond to the 
usual Mermin argument for standard graph states. Examples of 
solutions for
hypergraph states are $(7,4)$, $(15,4)$, or $(15,8)$, see the lower ends of
the red 
lines. The structure of odd and even binomial coefficients (marked 
in blue and yellow, respectively) is known to approximate 
the fractal of the Sierpi\'nski triangle. From this one
directly finds an arbitrary number of noncontextuality 
inequalities [e.g., for  $(N,k)=(31,16)$]. The figure
was generated with code from Ref.~\cite{figure}. 
}
\label{triangle}
\end{figure}

There are several conditions which have to be fulfilled by the stabilizer
generated by $\{g_i\}$ in order to fulfill Eq.~(\ref{ansatzghz}):

\begin{enumerate}

\item First, if we want to apply the same trick as in the 
GHZ argument, the product of the first terms on the right-hand 
side of Eq.~(\ref{ansatzghz}) has to equal the last term, up 
to a sign. This means that in this product each  $C_{e\setminus\{i\}}$ 
has to occur an even number of times. For $k$-uniform fully connected 
hypergraphs each of the $C_{e\setminus\{i\}}$ occurs $N-(k-1)$ times, 
so this number has to be even.

\item 
For the last term on the right-hand side of Eq.~(\ref{ansatzghz}) 
there has to be a minus sign. We write
$\prod_{i=1}^N g_i = (\prod_{e \in E} C_e) (\prod_{i=1}^N X_i) (\prod_{e \in E} C_e)$
and with the help of Lemma~2 
one finds that pulling each $C_e$ 
from the left to the right results in a minus sign. It follows that the total number 
of $C_e$ must be odd, so ${N \choose k}$ must be odd.

\item Finally, when applying Lemma 2 to  the term 
$\prod_{i=1}^N g_i$, all the further $C_j$ for $j \neq \emptyset$ 
and $j \neq e$ have to cancel. If $j=\{j_1\}$  consists of a 
single element, the term $C_j$ arises from a $C_e$ if $j_1 \in e$. 
For fully-connected $k$-uniform hypergraphs, this happens 
${N-1 \choose k-1}$ times, so this number must be even. By
considering $j=\{j_1, j_2\}$ one finds that ${N-2 \choose k-2}$ 
must be even. In general, any ${N-\alpha \choose k-\alpha}$ must 
be even for any $\alpha < k$. Note that for  $\alpha=k-1$ this 
implies the first condition explained above.
\end{enumerate}

Considering the possible values for $N$ and $k$ fulfilling the conditions
above, one finds the 
first possible choice to be $(N,k)=(3,2)$, which is the original 
Mermin inequality (see Fig.~\ref{triangle}). Further examples are 
$(N,k)=(3+4r,2)$. These are, however, just known Bell inequalities 
for usual graph states. The first example for a hypergraph is
the $(N,k)=(7,4)$, i.e.\ a 7-qubit 4-uniform fully connected hypergraph
state. This leads to the operator 
\begin{align}
\MM & = \sum_{i=1}^{7} g_i  + \prod_{i=1}^7 g_i 
= 
\sum_{i=1}^{7} X_i \big(\prod_{e \in E(i)} C_{e\setminus\{i\}}\big) -
\prod_{i=1}^7 X_i.
\label{ghz74}
\end{align}
Note that here each $(\prod_{e \in E(i)} C_{e\setminus\{i\}})$ 
contains 20 operators $C_f$. For any classical model which assigns 
to the $X_i$ and $C_{e\setminus\{i\}}$ the values $\pm 1$, we have 
$\mean{\MM} \leq 6$, but for the corresponding hypergraph state, 
we have $\mean{\MM} = 8.$ This shows that the hypergraph formalism 
can be used to develop novel Kochen-Specker inequalities. 

For an experimental test of this inequality, the experimenter should
measure each of the eight correlation terms from Eq.~(\ref{ghz74}). For
each of them, he has to prepare the corresponding hypergraph
state and then measure jointly or in a sequence the observables $X_i$ 
and $C_{e\setminus\{i\}}$, which is possible, as these observables 
commute. The results are multiplied and averaged over many repetitions
of the experiment, to compute $\mean{\MM}$.

Two remarks should be added here. First, it is remarkable that
the conditions to obtain a noncontextuality argument lead to 
the fractal structure of the Sierpi\'nski triangle (see Fig.~\ref{triangle}) 
and from this one can directly read of an infinite number of further 
noncontextuality arguments. In fact, one can directly see that if we 
take an arbitrary $r \geq 1$ and $s \geq 0$ and choose
\be
k= 2^{r} \mbox{ and } N=(2^{r+1}-1) + 2 s k,
\ee
the corresponding $(N, k)$ fully-connected $k$-uniform hypergraph 
state can be used to derive a noncontextuality inequality in the same way as
before.

Second, we introduce for the case $(N,k)=(7,4)$ the observables
\be
A_i = \prod_{e \;{\rm with } \; i \in e \; {\rm and } \; |e|=3 } C_{e}
\ee
as the product of all possible $C_e$ where $e$ are 3-edges 
that contain $i$. Then we have 
$A_2 A_3 A_4 A_5 A_6 A_7 =\prod_{e \in E(1)} C_{e\setminus\{1\}}$,
so the product of the six $A_i$ contains all $C_e$, where $e$ 
are 3-edges that do {\it not} contain $i=1$. Then we can write
$g_1 = X_1 A_2 A_3 A_4 A_5 A_6 A_7$ and the noncontextuality 
inequality takes the form
\be
\mean{\MM} = \mean{X_1 A_2 A_3 A_4 A_5 A_6 A_7} + \mbox{ permutations }
- \mean{X_1 X_2 X_3 X_4 X_5 X_6 X_7} \leq 6,
\ee
while for the hypergraph state we have $\mean{\MM}=8.$ This form shows
the close relation to the original Mermin inequality in 
Eq.~(\ref{mermin2}). Note, however, that this is only a formal similarity
and that in this form the $A_i$ are not the appropriate observables to 
measure, since $X_i$ and $A_j$ for $i\neq j$ do not commute.

\subsection{Bell inequalities in the hypergraph formalism}

So far, we have shown how the hypergraph formalism can be used 
to derive noncontextuality inequalities for nonlocal observables. 
It would be very desirable, however, to find also Bell inequalities
for local observables. To demonstrate that this is indeed possible, 
one can decompose the nonlocal observables into local ones, and can 
try to show that a local assignment of classical values for these 
local observables cannot lead to the same result as quantum physics.

To start, let us consider Eq.~(\ref{ghz74}). We can write the first
term in $\MM$ as
\begin{align}
g_1 & = X_1 \otimes \big(\prod_{e \in E(1)} C_{e\setminus\{1\}}\big) 
\nonumber
\\
&=\frac{1}{128}\big[
X_1 \otimes \big(
48 \cdot \eins^{\otimes{6}}
+16 \cdot [Z_2 Z_3 \eins \eins \eins \eins + \mbox{ permutations }]
\nonumber
\\
&
-16 \cdot [Z_2 Z_3 Z_4 Z_5 \eins \eins + \mbox{ permutations }]
+ 80 \cdot Z^{\otimes{6}}
\big)
\big],
\end{align}
and the other terms can be represented similarly. In this way, we can express the
correlation $\MM$ only in terms of local measurements $X_i$ and $Z_i$. The question 
is, whether the bound $\mean{\MM}\leq 6$ still holds, if we assign to these 
local observables the values $\pm 1$ in a local manner. 

This is indeed the case and it can be seen as follows: Any local assignment 
of $\pm 1$ to three observables $Z_i, Z_j,$ and $Z_k$ gives a value for the 
nonlocal observable $C_{\{i,j,k\}}$ as this can be expressed as a classical
function in terms of the $Z_i, Z_j,$ and $Z_k$.
Moreover, one can directly check that from a local assignmen to
$Z_i, Z_j,$ and $Z_k$ the observable $C_{\{i,j,k\}}$ can only take the 
values $\pm 1$. Therefore, a local assignment to the $Z_i$ leads to a 
noncontextual assignment of $C_{e\setminus \{i\}}$ in Eq.~(\ref{ghz74})
and the bound $\mean{\MM}\leq 6$ holds. In this way, all the noncontextuality inequalities from the previous section can be interpreted as Bell inequalities, 
if the nonlocal observables $C_e$ are decomposed into local observables~$Z_i$.

Finally, we note that it is not clear how to construct a Bell or 
noncontextuality inequality with the help of the stabilizer for 
arbitrary hypergraph states.  The simple method for graph states 
from Ref.~\cite{GTHB04} is not successful, since multiplying the elements
of the stabilizer of the hypergraph state does not necessarily result in 
different signs. We believe that finding more general Bell inequalities 
for hypergraph states is a topic worth of further investigation.

\section{Conclusions}
In conclusion we have investigated the entanglement properties 
and nonclassical features of hypergraph states. We have characterized
the local unitary equivalence classes up to four-qubits and have provided general
conditions under which the local unitary equivalence of hypergraph states can simply
be decided by considering the finite set of local Pauli transformations. 
Finally, we have shown that the stabilizer formalism of hypergraph states can be
used to derive various inequalities for testing the Kochen-Specker theorem or Bell's
theorem.

There are many questions which are worth to be addressed in the future. First, 
it would be highly desirable to find applications of hypergraph states, e.g. 
for quantum error correction. Second, it would be interesting to see whether 
the stabilizer formalism of hypergraph states can also be used to derive
Bell inequalities with an increasing violation. Finally, proposals for the
experimental generation of hypergraph states with photons or ions are still 
missing.

\begin{table}[h]
\begin{center}
\begin{tabular}{|c|l|}
\hline
Class & Edges 
\\
\hline
\hline
1 & \scriptsize \{\{1, 2, 3\}, \{1, 2, 4\}, \{1, 2, 5\}, \{1, 2, 6\}, \{1, 3, 4\}, \{1, 3, 5\}, \{1,
   3, 6\}, \{2, 3, 4\}, \{2, 3, 5\}, 
   \{2, 3, 6\}\}
\\[-0.1cm]
2 & \scriptsize \{\{1, 2, 3\}, \{1, 2, 4\}, \{1, 2, 5\}, \{1, 2, 6\}, \{1, 3, 4\}, \{1, 3, 5\}, \{1,
   4, 6\}, \{2, 3, 4\}, \{2, 3, 5\},  \{2, 4, 6\}\}
\\[-0.1cm]
3 & \scriptsize  \{\{1, 2, 3\}, \{1, 2, 4\}, \{1, 3, 5\}, \{1, 4, 6\}, \{1, 5, 6\}, \{2, 3, 6\}, \{2,
   4, 5\}, \{2, 5, 6\}, \{3, 4, 5\}, \{3, 4, 6\}\}
\\[-0.1cm]
4 & \scriptsize  \{\{1, 2, 3\}, \{1, 2, 4\}, \{1, 2, 5\}, \{1, 2, 6\}, \{1, 3, 4\}, \{1, 4, 5\}, \{1,
   4, 6\}, \{2, 3, 4\}, \{2, 4, 5\},  \{2, 4, 6\}, \\[-0.2cm]
   & \scriptsize\{3, 5, 6\}\}
\\[-0.1cm]
5 & \scriptsize  \{\{1, 2, 3\}, \{1, 2, 4\}, \{1, 2, 5\}, \{1, 2, 6\}, \{1, 3, 4\}, \{1, 3, 5\}, \{1,
   3, 6\}, \{1, 4, 5\}, \{2, 3, 4\},  \{2, 3, 5\},\\[-0.2cm]
   & \scriptsize \{2, 3, 6\}, \{2, 4, 5\}\}
\\[-0.1cm]
6 & \scriptsize  \{\{1, 2, 3\}, \{1, 2, 4\}, \{1, 2, 5\}, \{1, 2, 6\}, \{1, 3, 4\}, \{1, 3, 5\}, \{1,
   4, 5\}, \{1, 4, 6\}, \{2, 3, 4\}, \{2, 4, 5\}, \\[-0.2cm]
   & \scriptsize \{2, 4, 6\}, \{3, 5, 6\}\}
\\[-0.1cm]
7 &  \scriptsize \{\{1, 2, 3\}, \{1, 2, 4\}, \{1, 2, 5\}, \{1, 2, 6\}, \{1, 3, 4\}, \{1, 3, 5\}, \{1,
   4, 6\}, \{1, 5, 6\}, \{2, 3, 4\},  \{2, 3, 5\},\\[-0.2cm]
   & \scriptsize \{2, 4, 6\}, \{2, 5, 6\}\}
\\[-0.1cm]
8 & \scriptsize  \{\{1, 2, 3\}, \{1, 2, 4\}, \{1, 2, 5\}, \{1, 2, 6\}, \{1, 3, 4\}, \{1, 3, 5\}, \{1,
   4, 6\}, \{1, 5, 6\}, \{2, 3, 4\},  \{2, 3, 5\},\\[-0.2cm]
   & \scriptsize \{2, 4, 6\}, \{3, 5, 6\}\}
\\[-0.1cm]
9 & \scriptsize  \{\{1, 2, 3\}, \{1, 2, 4\}, \{1, 2, 5\}, \{1, 2, 6\}, \{1, 3, 4\}, \{1, 4, 5\}, \{1,
   4, 6\}, \{2, 3, 4\}, \{2, 4, 5\},  \{2, 5, 6\}, \\[-0.2cm]
   & \scriptsize\{3, 4, 6\}, \{3, 5, 6\}\}
\\[-0.1cm]
10 & \scriptsize  \{\{1, 2, 3\}, \{1, 2, 4\}, \{1, 2, 5\}, \{1, 2, 6\}, \{1, 3, 4\}, \{1, 4, 5\}, \{1,
   5, 6\}, \{2, 3, 4\}, \{2, 4, 5\}, \{2, 5, 6\},\\[-0.2cm]
   & \scriptsize \{3, 4, 6\}, \{3, 5, 6\}\}
\\[-0.1cm]
11 & \scriptsize \{\{1, 2, 3\}, \{1, 2, 4\}, \{1, 2, 5\}, \{1, 2, 6\}, \{1, 3, 4\}, \{1, 3, 5\}, \{1,
   4, 5\}, \{1, 4, 6\}, \{2, 3, 4\}, \{2, 3, 5\}, \\[-0.2cm]
   & \scriptsize\{2, 4, 5\}, \{2, 4, 6\}, \{3, 
  5, 6\}\}
\\[-0.1cm]
12 & \scriptsize  \{\{1, 2, 3\}, \{1, 2, 4\}, \{1, 2, 5\}, \{1, 2, 6\}, \{1, 3, 4\}, \{1, 3, 5\}, \{1,
   4, 6\}, \{1, 5, 6\}, \{2, 3, 4\}, \{2, 4, 6\},\\[-0.2cm]
   & \scriptsize \{2, 5, 6\}, \{3, 4, 5\}, \{3, 
  5, 6\}\}
\\[-0.1cm]
13 & \scriptsize  \{\{1, 2, 3\}, \{1, 2, 4\}, \{1, 2, 5\}, \{1, 2, 6\}, \{1, 3, 4\}, \{1, 3, 5\}, \{1,
   3, 6\}, \{1, 4, 5\}, \{2, 3, 4\}, \{2, 3, 6\}, \\[-0.2cm]
   & \scriptsize\{2, 4, 5\}, \{2, 4, 6\}, \{3, 
  4, 5\}, \{3, 4, 6\}\}
\\[-0.1cm]
14 &  \scriptsize \{\{1, 2, 3\}, \{1, 2, 4\}, \{1, 2, 5\}, \{1, 2, 6\}, \{1, 3, 4\}, \{1, 3, 5\}, \{1,
   4, 5\}, \{1, 4, 6\}, \{2, 3, 4\}, \{2, 3, 5\},\\[-0.2cm]
   & \scriptsize \{2, 4, 5\}, \{2, 4, 6\}, \{3, 
  4, 5\}, \{3, 5, 6\}\}
\\[-0.1cm]
15 & \scriptsize  \{\{1, 2, 3\}, \{1, 2, 4\}, \{1, 2, 5\}, \{1, 2, 6\}, \{1, 3, 4\}, \{1, 3, 5\}, \{1,
   4, 5\}, \{1, 4, 6\}, \{2, 3, 4\}, \{2, 3, 5\},\\[-0.2cm]
   & \scriptsize \{2, 4, 6\}, \{2, 5, 6\}, \{3, 
  4, 5\}, \{3, 5, 6\}\}
\\[-0.1cm]
16 & \scriptsize  \{\{1, 2, 3\}, \{1, 2, 4\}, \{1, 2, 5\}, \{1, 2, 6\}, \{1, 3, 4\}, \{1, 3, 5\}, \{1,
   4, 5\}, \{1, 4, 6\}, \{2, 3, 4\},  \{2, 3, 5\},\\[-0.2cm]
   & \scriptsize\{2, 5, 6\}, \{3, 4, 5\}, \{3, 
  4, 6\}, \{3, 5, 6\}\}
\\[-0.1cm]
17 & \scriptsize  \{\{1, 2, 3\}, \{1, 2, 4\}, \{1, 2, 5\}, \{1, 2, 6\}, \{1, 3, 4\}, \{1, 3, 5\}, \{1,
   4, 6\}, \{1, 5, 6\}, \{2, 3, 4\}, \{2, 3, 5\}, \\[-0.2cm]
   & \scriptsize\{2, 4, 6\}, \{2, 5, 6\}, \{3, 
  4, 5\}, \{3, 4, 6\}\}
\\[-0.1cm]
18 & \scriptsize \{\{1, 2, 4\}, \{1, 2, 5\}, \{1, 2, 6\}, \{1, 3, 4\}, \{1, 3, 5\}, \{1, 4, 5\}, \{1,
   4, 6\}, \{2, 3, 5\}, \{2, 3, 6\},  \{2, 4, 6\}, \\[-0.2cm]
   & \scriptsize\{2, 5, 6\}, \{3, 4, 5\}, \{3, 
  4, 6\}, \{3, 5, 6\}\}
\\[-0.1cm]
19 &  \scriptsize \{\{1, 2, 3\}, \{1, 2, 4\}, \{1, 2, 5\}, \{1, 2, 6\}, \{1, 3, 4\}, \{1, 3, 5\}, \{1,
   4, 5\}, \{1, 4, 6\}, \{2, 3, 4\},  \{2, 3, 5\}, \\[-0.2cm]
   & \scriptsize\{2, 4, 5\}, \{2, 5, 6\}, \{3, 
  4, 5\}, \{3, 4, 6\}, \{3, 5, 6\}\}
\\[-0.1cm]
20 & \scriptsize  \{\{1, 2, 3\}, \{1, 2, 4\}, \{1, 2, 5\}, \{1, 2, 6\}, \{1, 3, 4\}, \{1, 3, 5\}, \{1,
   4, 5\}, \{1, 4, 6\}, \{2, 3, 5\},  \{2, 3, 6\},\\[-0.2cm]
   & \scriptsize \{2, 4, 5\}, \{2, 4, 6\}, \{3, 
  4, 5\}, \{3, 4, 6\}, \{3, 5, 6\}\}
\\[-0.1cm]
21 &  \scriptsize \{\{1, 2, 3\}, \{1, 2, 4\}, \{1, 2, 5\}, \{1, 2, 6\}, \{1, 3, 4\}, \{1, 3, 5\}, \{1,
   4, 5\}, \{1, 4, 6\}, \{2, 3, 5\},  \{2, 3, 6\}, \\[-0.2cm]
   & \scriptsize\{2, 4, 6\}, \{2, 5, 6\}, \{3, 
  4, 5\}, \{3, 4, 6\}, \{3, 5, 6\}\}
\\[-0.1cm]
22 &  \scriptsize \{\{1, 2, 3\}, \{1, 2, 4\}, \{1, 2, 5\}, \{1, 3, 4\}, \{1, 3, 6\}, \{1, 4, 5\}, \{1,
   4, 6\}, \{1, 5, 6\}, \{2, 3, 5\},  \{2, 3, 6\}, \\[-0.2cm]
   & \scriptsize\{2, 4, 5\}, \{2, 4, 6\}, \{2, 
  5, 6\}, \{3, 4, 6\}, \{3, 5, 6\}\}
\\[-0.1cm]
23 & \scriptsize \{\{1, 2, 3\}, \{1, 2, 4\}, \{1, 2, 5\}, \{1, 3, 4\}, \{1, 3, 6\}, \{1, 4, 5\}, \{1,
   4, 6\}, \{1, 5, 6\}, \{2, 3, 5\},  \{2, 3, 6\},\\[-0.2cm]
   & \scriptsize \{2, 4, 5\}, \{2, 4, 6\}, \{2, 
  5, 6\}, \{3, 4, 5\}, \{3, 4, 6\}, \{3, 5, 6\}\}
\\[-0.1cm]
24 &  \scriptsize \{\{1, 2, 4\}, \{1, 2, 5\}, \{1, 2, 6\}, \{1, 3, 4\}, \{1, 3, 5\}, \{1, 3, 6\}, \{1,
    4, 5\}, \{1, 4, 6\}, \{2, 3, 5\}, \{2, 3, 6\},\\[-0.2cm]
   & \scriptsize \{2, 4, 5\}, \{2, 4, 6\}, \{2, 
   5, 6\}, \{3, 4, 5\}, \{3, 4, 6\}, \{3, 5, 6\}\}
\\
\hline
\end{tabular}
\end{center}
\caption{The 24 classes of three-uniform six-qubit hypergraph states, where all the reduced
single particle density matrices are maximally mixed.
}
\label{6qbtable}
\end{table}

\section*{Acknowledgments}
We thank  Markus Grassl, Marcus Huber, Laura Man\v{c}inska and Andreas Winter for discussions.
This work  
has been supported by the EU (Marie Curie CIG 293993/ENFOQI),
the BMBF (Chist-Era Project QUASAR), the FQXi Fund
(Silicon Valley Community Foundation), the Austrian Science Fund 
(FWF) Grant No.~Y535-N16, the program ``Science without borders'' from
the Brazilian Agency CAPES, and the DFG.

\section{Appendix}
In Table \ref{6qbtable} we present the 24 classes of three-uniform six-qubit 
hypergraph states, where all single-qubit reduced density matrices are maximally
mixed.

\end{document}